\documentclass[letterpaper, 10 pt, conference]{ieeeconf}  

\IEEEoverridecommandlockouts                              

\UseRawInputEncoding

\usepackage{graphics, subfig} 
\usepackage{epsfig} 
\usepackage{mathptmx} 
\usepackage{times} 
\usepackage{amsmath} 
\usepackage{amssymb}  
\usepackage[ruled,vlined]{algorithm2e}
\usepackage{bbm}
\usepackage{comment}
\usepackage{multirow}
\usepackage{theorem}

\newtheorem{remark}{Remark}
\newtheorem{lemma}{Lemma}
\newtheorem{cor}{Corollary}
\newtheorem{proposition}{Proposition}
\newtheorem{theorem}{Theorem}

\def\qedsymbol{\ensuremath{\Box}}      

\def\qed{\ifhmode\unskip\nobreak\fi\quad\qedsymbol}     
\def\frqed{\ifhmode\nobreak\hbox to5pt{\hfil}\nobreak%
	\hskip 0pt plus1fill\nobreak\fi\quad\qedsymbol\renewcommand{\qed}{}} 

\def\QEDsymbol{\vrule width.6em height.5em depth.1em\relax}
\def\frQED{\ifhmode\nobreak\hbox to5pt{\hfil}\nobreak%
	\hskip 0pt plus1fill\nobreak\fi\quad\QEDsymbol\renewcommand{\qed}{}} 
\def\QED{\ifhmode\unskip\nobreak\fi\quad\QEDsymbol}     
\newtheorem{assumption}{Assumption}
\newcommand{\sdb}[1]{\textcolor{red}{#1}}

\title{\LARGE \bf
\texttt{FlipDyn}: A game of resource takeovers in dynamical systems
}

\author{Sandeep~Banik and~Shaunak~D.~Bopardikar
	\thanks{The authors are with the Department of Electrical and Computer Engineering at Michigan State University, East Lansing, MI, USA. Emails: \texttt{baniksan@msu.edu; shaunak@egr.msu.edu}}}

\begin{document}

\maketitle
\thispagestyle{empty}
\pagestyle{empty}

\begin{abstract}
We introduce a game in which two players with opposing objectives seek to repeatedly takeover a common resource. The resource is modeled as a discrete time dynamical system over which a player can gain control after spending a state-dependent amount of energy at each time step. We use a FlipIT-inspired deterministic model that decides which player is in control at every time step. A player's policy is the  probability with which the player should spend energy to gain control at each time step. Our main results are three-fold. First, we present analytic expressions for the cost-to-go as a function of the hybrid state of the system, i.e., the physical state of the dynamical system and the binary \texttt{FlipDyn} state for any general system with arbitrary costs. These expressions are exact when the physical state is also discrete and has finite cardinality. Second, for a continuous physical state with linear dynamics and quadratic costs, we derive expressions for Nash equilibrium (NE). For scalar physical states, we show that the NE depends only on the parameters of the value function and costs, and is independent of the state. Third, we derive an approximate value function for higher dimensional linear systems with quadratic costs. Finally, we illustrate our results through a numerical study on the problem of controlling a linear system in a given environment in the presence of an adversary. 
\end{abstract}

\section{Introduction}
Rising automation, inexpensive computation and  proliferation of the Internet of Things have made cyber-physical systems (CPS) ubiquitous in industrial control systems, home automation, autonomous vehicles, smart grids and medical devices~\cite{rajkumar2010cyber,baheti2011cyber}. However, increased levels of connectivity and ease of operations also makes CPS vulnerable to cyber and physical attacks~\cite{cardenas2008research,parkinson2017cyber}. An adversarial takeover can drive the system to undesirable states or can even permanently damage the system causing disruption in services and potential loss of lives. Therefore, it becomes imperative to develop policies to continuously scan for adversarial behavior while striking a balance between operating costs and system integrity. This paper proposes an approach to model and analyze the problem of resource takeovers in CPS.   

\smallskip

Consider an adversary that has access to a prototypical CPS control loop and can achieve a  takeover at various points shown in Figure~\ref{fig:FlipIt_general_dynamics}. These points include the (i) reference inputs, (ii) actuator,  (iii) state, (iv) sensor and (v) control output, thereby affecting the system performance. As opposed to conventional adversaries perturbing the states of the system (actuator attack) or measurements (integrity attack)~\cite{fawzi2012security}, this work supposes that an adversary completely takes over a resource and can transmit arbitrary values originating from the controlled resource. 


There has been a lot of recent research into CPS security in the controls community. The work~\cite{bai2015security} focuses on resilience against an adversary who can hijack and replace the control signal while remaining undetected to minimize the control performance. This idea is generalized in~\cite{bai2017data,katewa2021detection} for any linear stochastic system to determine its detectability, while quantifying performance degradation. 
Reference~\cite{amin2013cyber} developed model-based observers to detect and isolate such stealthy deception attacks to make water SCADA systems resilient. The authors in~\cite{chang2018secure} developed a secure estimator in conjunction with a Kalman filter for CPS where the set of attacked sensor can change over time. 

\smallskip

Game theory has also been extensively applied to model CPS security problems. A two-player nonzero sum game with asymmetric information and resource constraints between a controller and a jammer was introduced in~\cite{gupta2016dynamic}. Reference~\cite{chen2017security} studies contract design at the physical layer to ensure cloud security quality of service. Similarly, a two-player dynamic game between a network designer and an adversary is used to determine policies to keep the infrastructure networks of a CPS protected and enable recovery under an attack~\cite{chen2019dynamic}. A range of works in designing physical and cyber security policies using game-theoretical frameworks are presented in~\cite{zhu2015game}. A non-cooperative game between a defender (contractive controller) and an adversary (expanding controller) was presented in~\cite{kontouras2014adversary}, limited to finite and fixed periods of control by each player. Covert attacks driving the system states outside the performance set competing against a contractive control subject to control and state constraints were studied in~\cite{kontouras2015covert}. 

\smallskip

The setup in this paper is inspired by the cybersecurity game of stealthy takeover known as FlipIt~\cite{van2013flipit}. FlipIT is a two-player game between an adversary and defender competing to control a shared resource. The resource can be represented as a critical digital system such as a computing device, virtual machines or a cloud service~\cite{bowers2012defending}. The authors in~\cite{laszka2014flipthem} extend the FlipIT model to a more general framework of multiple resource takeovers, termed as FlipThem, where an attacker has to compromise all the resources or only one in order to take over the entire system. The FlipIT model has also been applied to supervisory control and data acquisition (SCADA)~\cite{liu2021flipit} system, commonly used in industrial automation, by deriving the probability distribution of time-to-compromise the system and evaluating the impacts of insider assistance for an adversary. Largely, the FlipIT setups have been limited to a \emph{static system}, i.e, the payoff does not change over time, whereas in this work, we model the takeover of a \emph{dynamical} system between an adversary and a defender.

In this paper, we assume that a controller is already given for both the defender and an adversary. What is not known are the time instants at which each player should act to takeover the system. Thus, our set-up generalizes the formulation considered in \cite{kontouras2014adversary,kontouras2015covert} by explicitly attaching state-dependent costs on each player (controller). 


\begin{figure*}[h]
	\begin{center}
		\subfloat[]{\includegraphics[width = 0.28\linewidth]{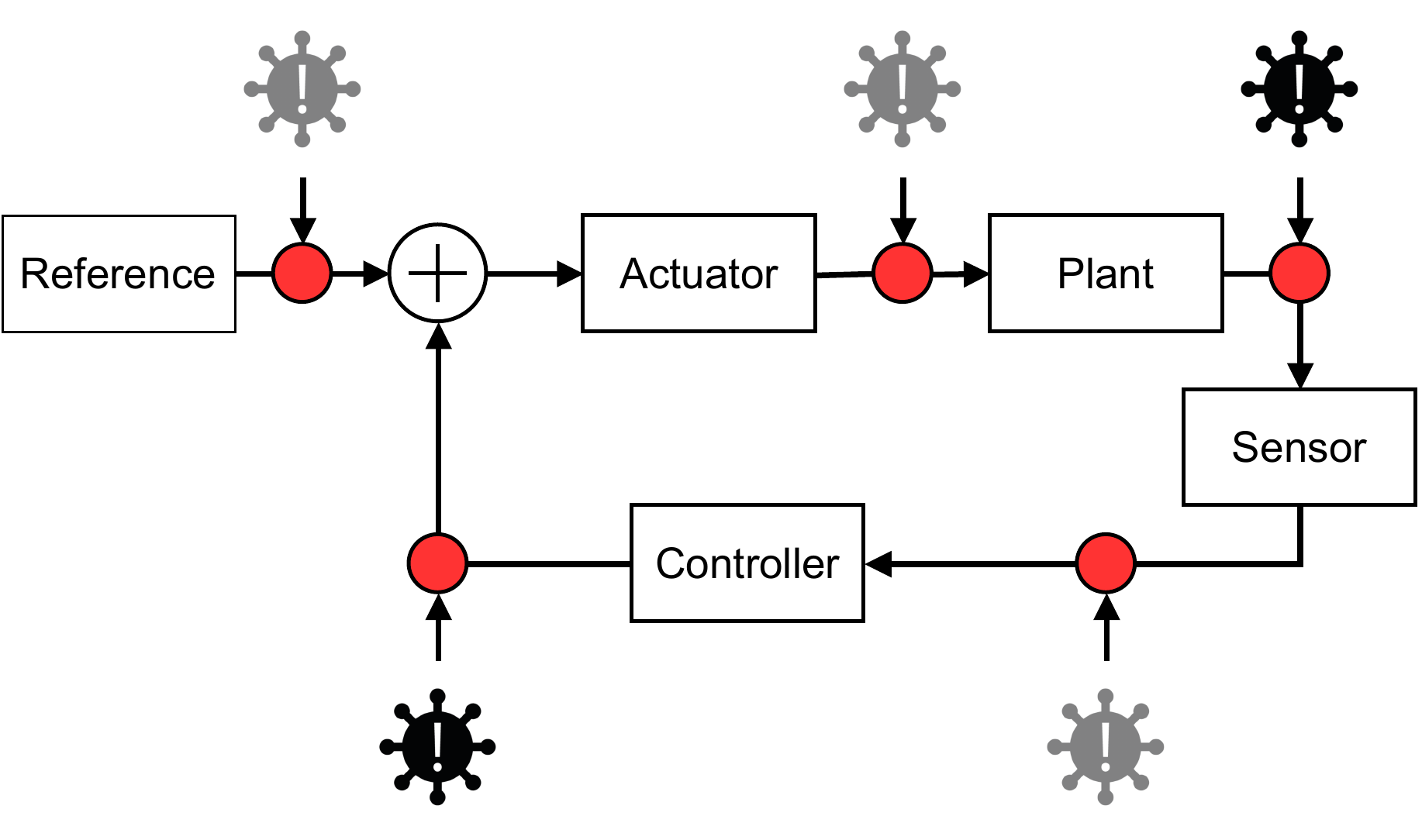}
			\label{fig:FlipIt_general_dynamics}	
		}
		\subfloat[]{\includegraphics[width = 0.30\linewidth]{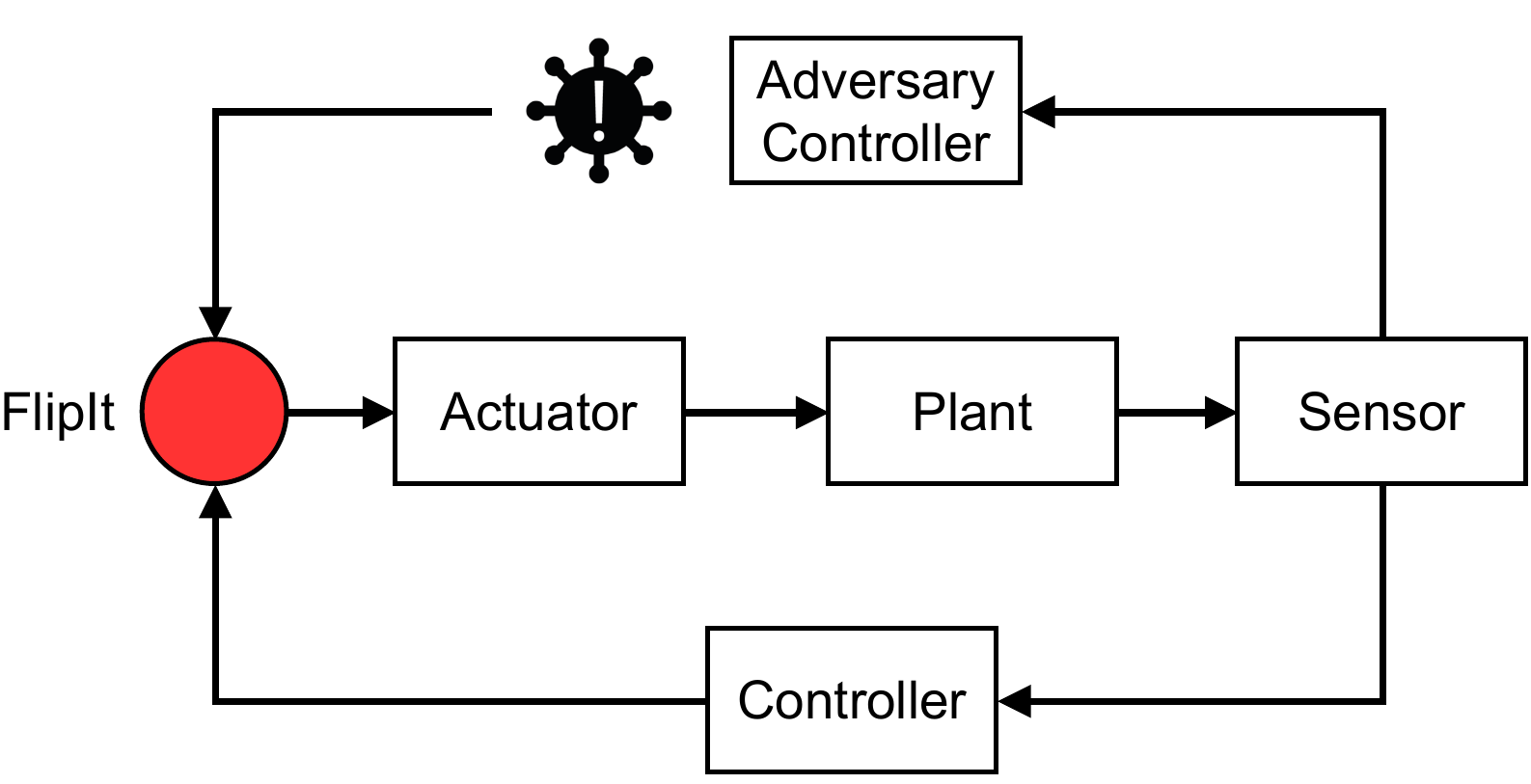}
			\label{fig:FlipIt_control_dynamics}	
		}
		\subfloat[]{\includegraphics[width = 0.36\linewidth]{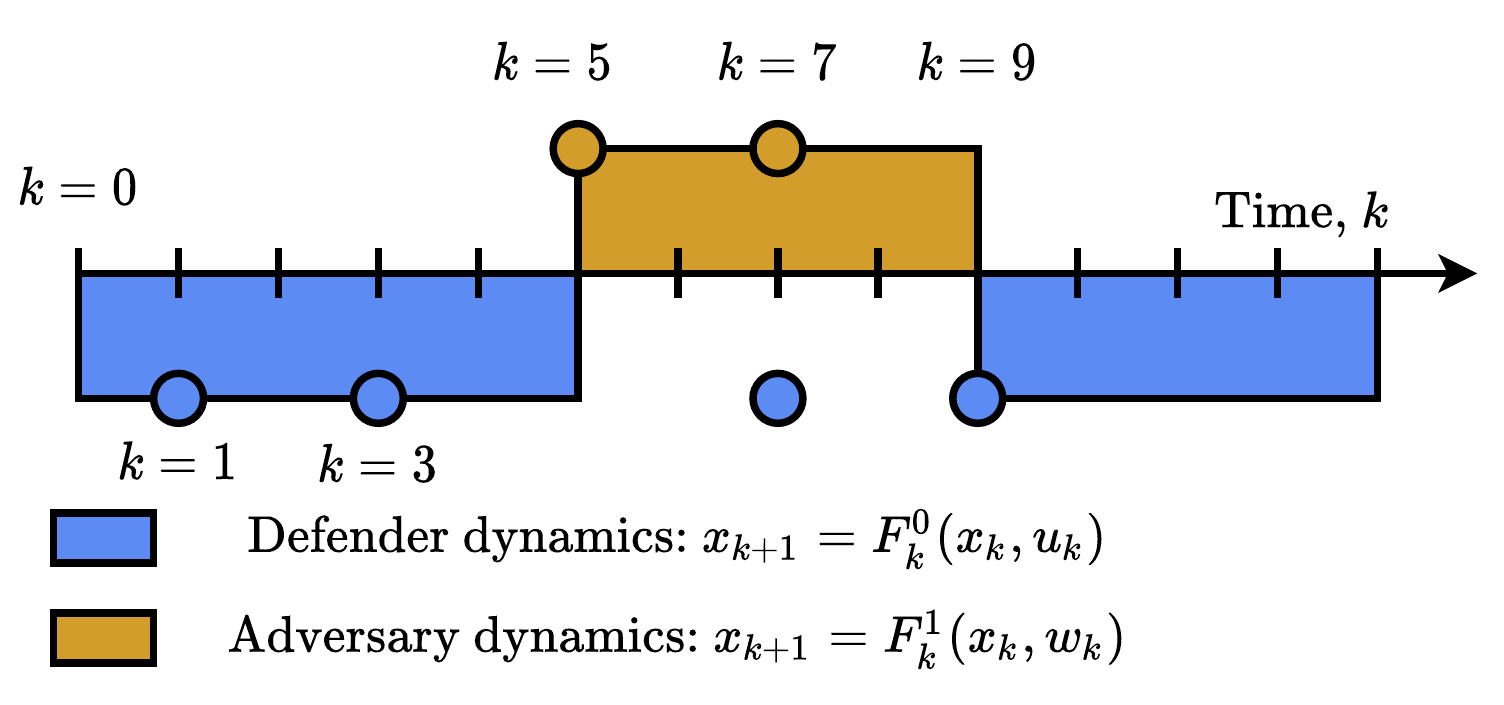}
		\label{fig:FlipDyn_sample}
		}
		\caption{\small (a) Closed-loop system with adversaries present at various locations infecting the reference values, actuator, plant, measurement output and control input. (b) Closed-loop system with the adversary present between the controller and actuator trying to takeover the control signals. A FlipIt is setup over the control signal between the defender and adversarial control. (c) Sample sequence of the \texttt{FlipDyn} game with the defender action and takeover indicated by the blue circles and region, respectively. Similarly, the adversary action and takeover time period are indicated by the red circles and region, respectively.}	
		\label{fig:FlipIt_dynamics}
	\end{center}
	\vspace{-0.25in}
\end{figure*} 
The contributions of this paper are three-fold.

\textbf{1. Game-theoretic modeling of dynamic resource takeover:} We model a two-player zero-sum game between a defender and an adversary trying to takeover a dynamical system (resource). We term this as the \texttt{FlipDyn} game. Given the controllers used by each player, this model accounts for state-dependent takeover costs subject to the system dynamics when controlled by either player.

\textbf{2. FlipDyn control for any general system:} We characterize the Nash equilibrium (NE) of the \texttt{FlipDyn} game as a function of both the continuous state of the system and the binary \texttt{FlipDyn} state. For finite cardinality of the physical state and arbitrary takeover and stage costs, we obtain the exact value and corresponding policies for the game. 

\textbf{3. NE for linear dynamical systems with quadratic costs:} We derive the NE for a linear dynamical system with takeover and stage costs that are quadratic in the state. For scalar systems, we show that the solutions are a function of only the parameters of the system dynamics and costs. For higher dimensional systems, we provide an approximate solution for the value of the \texttt{FlipDyn} game. We illustrate our findings through two numerical examples.  
\medskip

The paper is organized as follows. In Section~\ref{sec:Problem_Formulation}, we formally define the \texttt{FlipDyn} game. We provide a solution methodology for any general system with arbitrary state and takeover costs in Section~\ref{sec:FlipDyn_General_Systems}. We present the analysis for linear systems with quadratic costs in Section~\ref{sec:FlipDyn_Linear_Systems}. In Section~\ref{sec:Numerics}, we illustrate the efficacy of the solution applied to a linear-time invariant system. We conclude this paper and provide future directions in Section~\ref{sec:Conclusion}.

\section{Problem Formulation}\label{sec:Problem_Formulation}
Consider a discrete-time dynamical system governed by 
\begin{align}\label{eq:dynamics}
	x_{k+1} = F_{k}^{0}(x_k,u_k),
\end{align}
where $k \in \mathbb{N}$ denotes the discrete time instant, $x_{k} \in \mathbb{R}^{n}$ and  $u_{k} \in \mathbb{R}^m$ are the state and control input of the system, respectively, $F_{k}^{0}: \mathbb{R}^{n} \times \mathbb{R}^{m} \rightarrow \mathbb{R}^{n}$ is the state transition function. We restrict our attention to a single adversary trying to gain control of the dynamical system resource~\eqref{eq:dynamics}. In particular, we assume the adversary to be located between the controller and actuator, illustrated in Figure~\ref{fig:FlipIt_control_dynamics}.
The inclusion of an adversary modifies~\eqref{eq:dynamics} resulting into 
\begin{align*}
	x_{k+1} = (1 - \alpha_{k})F_{k}^{0}(x_{k},u_{k}) + \alpha_{k}F_{k}^{1}(x_{k},w_{k}),
\end{align*}
where $F_{k}^{1}: \mathbb{R}^{n} \times \mathbb{R}^{p} \rightarrow \mathbb{R}^{n}$ is the state transition function under the adversary's control, $w_k \in \mathbb{R}^{p}$ represents the attack signal, and $\alpha_{k} \in \{0,1\}$ denotes a takeover of the control signal by either the adversary ($\alpha_{k} = 1$) or the defender ($\alpha_{k} = 0$), termed as the \texttt{FlipDyn} state. A takeover is  mutually exclusive, i.e., only one player is in control of the system at any given time.

The control law for each player are pre-designed with different objectives -- a defender's objective may be to steer the state towards an equilibrium point. In contrast the adversary upon gaining access into the system, implements an attack policy to ensure maximum divergence of the state from the corresponding equilibrium point, while keeping the state within any defined set. In particular, we assume that 
\[
u_k = K_k(x), \quad w_k = W_k(x),
\]
where $K_k$ and $W_k$ are specified state feedback control laws. These lead to the following closed-loop evolution 
\begin{align}\label{eq:dyn_sys_compact}
	x_{k+1} = (1 - \alpha_{k})f_{k}^{0}(x_{k}) + \alpha_{k}f_{k}^{1}(x_{k}),
\end{align}
where $f_{k}^{0}(x_{k}) := F_k^0(x_k, K_k(x))$ and $f_{k}^{1}(x_{k}) := F_k^1(x_k, W_k(x))$.

\medskip

To describe a takeover mathematically, the action $\alpha^{j}_k \in \{0,1\}$ denotes the $k$th move of the player $j \in \{0,1\}$, with $j = 0$ denoting the defender, and $j=1$ as the adversary. The dynamics of this binary \texttt{FlipDyn} state based on the player's move satisfies
\begin{align}\label{eq:flip_state}
	\alpha_{k} &= \begin{cases}
		\alpha_{k-1}, & \text{if } \alpha^{1}_{k} = \alpha^{0}_{k}, \\
		j, & \text{if } \alpha^{j}_{k} = 1.
	\end{cases}
\end{align}
Equation~\eqref{eq:flip_state} states that if both players act to obtain control of the resource at the same time, then their actions get nullified and \texttt{FlipDyn} state remains unchanged. However, if the resource is in control of one of the players and the other player moves to gain control at time $k+1$ while the first player does not exert control, then the \texttt{FlipDyn} state toggles. Finally, if a player is already in control and decides to move while the other player remains inactive, then the \texttt{FlipDyn} state is unchanged. 

A sample instance over a finite time period is demonstrated in Figure~\ref{fig:FlipDyn_sample}, where the defender has an initial access at time $k=0$, followed by a takeover action at time $k=1$ and $3$. The adversary gains access at time instant $k=5$ under a no defense action, and remains in control till time instant $k=9$, when the defender takes back control. Additionally, notice at time instant $k=7$, both the adversary and defender move to takeover, but their actions are cancelled out and therefore, the \texttt{FlipDyn} state does not change, i.e, the adversary maintains control. 

The joint state dynamics as a function of the binary \texttt{FlipDyn} state and the \texttt{FlipDyn} dynamics are described by \eqref{eq:dyn_sys_compact} and 
\begin{align}
	\label{eq:FlipDyn_compact}
	\alpha_{k+1} &= \left((1 - \alpha^{0}_k)(1 - \alpha^{1}_k) +  \alpha^{0}_k\alpha^{1}_k\right)\alpha_{k} + (1- \alpha^{0}_k)(\alpha^{0}_k + \alpha^{1}_k).
\end{align}
We pose the decision to control the resource as a zero-sum dynamic game described by the dynamics~\eqref{eq:dyn_sys_compact} and~\eqref{eq:FlipDyn_compact} over a finite time horizon of $L$, where the defender aims to minimize a net cost given by, 
\begin{equation}\label{eq:obj_def_alpha}
	\begin{aligned}
		J(x_{0}, \alpha_0, \{\alpha^1_{t}\}, \{\alpha^0_{t}\}) = \sum_{t=1}^{L} g(x_t) + (1 - \alpha_t)d(x_t) - \alpha_t a(x_t), \\
	\end{aligned}
\end{equation}
where $g(x_t): \mathbb{R}^{n} \rightarrow \mathbb{R}$ represents the state regulation cost, $d(x_t)$  and $a(x_t)$ are the instantaneous takeover costs for the defender and adversary, respectively. The notation $\{\alpha_t^j\} := \{\alpha_1^j, \dots, \alpha_L^j\}$. In contrast, the adversary aims to maximize the cost function~\eqref{eq:obj_def_alpha} leading to a zero-sum dynamic game, defining our \emph{\texttt{FlipDyn} game}.



\smallskip

We seek to find the NE of the game defined by \eqref{eq:obj_def_alpha}. However, a pure NE may not be guaranteed. For instance, a one-step horizon problem results into solving a $2\times 2$ matrix game, which need not admit a pure NE. To guarantee existence of NE, we will expand the set of player policies to behavioral policies --  probability distributions over the space of discrete actions at each time step~\cite{hespanha2017noncooperative}. 
Specifically, let,
\begin{equation}
	y_{k} = \begin{bmatrix}
		\beta_{k} & 1 - \beta_{k}
	\end{bmatrix}^{T}, \quad z_{k} = \begin{bmatrix}
	\gamma_{k} & 1 - \gamma_{k}
\end{bmatrix}^{T}
\end{equation}
be a behavioral policy for the defender and adversary at time instant $k$, such that $\beta_{k} \in [0,1], \gamma_{k} \in [0,1]$. Thus, $y_{k}, z_{k}  \in \Delta_{2}$, where $\Delta_{2}$ is the probability simplex in two dimensions. The cost~\eqref{eq:obj_def_alpha} then needs to be considered in expectation over the horizon.
Over the finite horizon $L$, let $y_{\mathbf{L}} = \{y_{1}, y_{2}, \dots, y_{L}\} \in \Delta^{L}_{2}$ and $z_{\mathbf{L}} = \{z_{1}, z_{2}, \dots, z_{L}\} \in \Delta^{L}_{2}$ be the sequence of defender and adversary behavioral policies. Then, the expected outcome of the \texttt{FlipDyn} game over the defined finite horizon $L$ is 
\begin{equation}
    \label{eq:opti_E_cost}
	J_{E}(x_0, \alpha_0, y_{\mathbf{L}}, z_{\mathbf{L}}) :=  \mathbb{E}_{y_{\mathbf{L}}, z_{\mathbf{L}}}[J( x_0, \alpha_0, \{\alpha^1_{t}\}, \{\alpha^0_{t}\})],
\end{equation}
where the expectation is computed using the distributions $y_{\mathbf{L}}$ and $z_{\mathbf{L}}$. Specifically, we seek a saddle-point solution ($y_{\mathbf{L}}^{*}, z_{\mathbf{L}}^{*}$) in the space of behavioral policies such that for all allowable $\forall (x_0, \alpha_0)$,
\begin{equation*}
	J_{E}(x_0, \alpha_0, y_{\mathbf{L}}^{*}, z_{\mathbf{L}}) \leq J_E(x_0, \alpha_0, y_{\mathbf{L}}^{*}, z_{\mathbf{L}}^{*}) \leq  J_{E}(x_0, \alpha_0, y_{\mathbf{L}}. z_{\mathbf{L}}^{*}), 
\end{equation*}
Together, the \texttt{FlipDyn} game is completely defined by the cost in \eqref{eq:opti_E_cost} subject to the dynamics in \eqref{eq:dyn_sys_compact} and \eqref{eq:FlipDyn_compact}.

\section{\texttt{FlipDyn} Control for general systems}\label{sec:FlipDyn_General_Systems}
In this section, we first compute NE for the \texttt{FlipDyn} game. We begin by defining the value function for the \texttt{FlipDyn} game.
\subsection{Value function}
Our approach is to define a separate value function in each of the two \texttt{FlipDyn} states. Let $V^{0}_k(x)$ and $V^{1}_k(x)$ be the two value functions in state $x$ at time instant $k$ corresponding to the \texttt{FlipDyn} state of $\alpha = 0$ and $1$, respectively. Then for $\alpha=0$, we have 
\begin{equation}
    \label{eq:V_k^0_cost_to_go}
	V^{0}_k(x) = g_k(x) + y_{k}^{T}\Xi^{0}_{k}z_{k}, 
\end{equation}
where $\Xi^{0}_{k} \in \mathbb{R}^{2 \times 2}$ is the cost-to-go matrix, and the actions of the defender (row player) and adversary (column player) applied on $\Xi^{0}_{k}$ returns the value corresponding to the state at time $k+1$. This instantaneous payoff matrix has the form  
\begin{equation}\label{eq:Cost_to_go}
\Xi_{k}^{0} = \begin{bmatrix}
		V_{k+1}^0(f^0_k(x)) & V_{k+1}^1(f_k^1(x)) - a(x) \\ 
		V_{k+1}^0(f^0_k(x)) + d(x) & V_{k+1}^0(f^0_k(x)) + d(x) - a(x)
	\end{bmatrix}.
\end{equation}
The matrix entries corresponding to $\Xi^{0}_{k}$ are determined by using the \texttt{FlipDyn} dynamics~\eqref{eq:dyn_sys_compact} and~\eqref{eq:FlipDyn_compact}. $\Xi_{k}^{0}(1,1)$ corresponds to both the defender and adversary staying idle, respectively. Similarly, $\Xi_{k}^{0}(2,2)$ corresponds to the action of takeover by both the defender and adversary. The off-diagonal entries are due to exactly one player taking control. We observe that the actions of the defender and adversary couple the value functions in each \texttt{FlipDyn} state $V^{0}_k$ and $V^{1}_k$. 

Following similar steps, the value function $V^{1}_{k}$ for the \texttt{FlipDyn} state $\alpha = 1$, and its corresponding cost-to-go matrix $\Xi_{k}^{1}$ is 
\begin{equation}
    \label{eq:V_k^1_cost_to_go}
	V^{1}_k(x) = g_k(x) + y_{k}^{T} \Xi_{k}^{1} z_{k},
\end{equation}
\begin{equation}
    \label{eq:zero_sum_cost_to_go_al1}
	\Xi_{k}^{1} =\begin{bmatrix}
		V_{k+1}^1\left(f^1_k(x)\right) & V_{k+1}^1\left(f_k^1(x)\right) - a(x) \\ 
		V_{k+1}^0\left(f^0_k(x)\right) + d(x) & V_{k+1}^1\left(f^1_k(x)\right) + d(x) - a(x)
	\end{bmatrix}.
\end{equation}
From the previous observation, both value functions $V^{0}_{k}$ and $V^{1}_{k}$ are coupled through the cost-to-go matrices $\Xi_{k}^{0}$ and $\Xi_{k}^{1}$.

\subsection{Expected Value of the \texttt{FlipDyn} game}
In each \texttt{FlipDyn} state ($\alpha = \{0,1\}$), the corresponding cost-to-go matrix defines a one-step zero-sum game with the defender aiming to minimize the value function, and the adversary trying to maximize the same. When a row or column domination~\cite{hespanha2017noncooperative} exists, it leads to a pure policy for at least one player. However, we first show that this game does not admit dominated policies in the following result.

\medskip
\begin{lemma}
	\label{lm:mixed_policy}
	For any $k \in \mathbb{N}$, there is no pure policy equilibrium for the one-step zero-sum games defined by the matrices $\Xi^{0}_{k}$ and $\Xi^{1}_{k}$.
\end{lemma}
\begin{proof}
	We prove the claim only for $\Xi^{0}_{k}$ since the conclusion for $\Xi^{1}_{k}$ is symmetric. Based on the entries in the first column of $\Xi_{k}^{0}$, we observe,
	\[ V_{k+1}^{0}(f_k^{0}(x)) \leq V_{k+1}^{0}(f_k^{0}(x)) + d(x),\]
	indicating that under a pure policy of staying idle, i.e., $z_{k} = \begin{bmatrix}
		1 & 0
	\end{bmatrix}^{T}$ by the adversary, the defender will prefer not to move, i.e., $y_{k} = \begin{bmatrix}
	1 & 0
	\end{bmatrix}^{T}$. Similarly, from the entry $\Xi^0_k(1,2)$, we infer that for an adversary to takeover the system, the conditions
	\[ V_{k}^{1}(f_k^{1}(x)) > V_{k}^{0}(f_k^{0}(x)) + d(x), \text{ and}\]
	\[ V_{k}^{1}(f_k^{1}(x)) > V_{k}^{0}(f_k^{0}(x)) + a(x)\]
	must hold, respectively. In conclusion, the condition of 
	\begin{equation}\label{eq:Lemma1}
	V_{k}^{1}(f_k^{1}(x)) > V_{k}^{0}(f_k^{0}(x)) + \max \{d(x),a(x)\},
	\end{equation}
	must hold. This proves that there cannot be any row or column domination in $\Xi_k^0$. Therefore, the defender and adversary will always mix between takeover (defend, attack) and staying idle. This completes the proof. 
\end{proof}
\medskip
 The analysis of Lemma~\ref{lm:mixed_policy}, particularly \eqref{eq:Lemma1} provides a condition for a mixed policy NE of the one-step game. Using this condition, we recursively derive the (mixed) value at any time instant $k$ for each binary \texttt{FlipDyn} state as summarized in Theorem~\ref{th:Cost_to_go_zerosum_value}. 
\begin{theorem}
    \label{th:Cost_to_go_zerosum_value}
    Given the cost-to-go matrices~\eqref{eq:Cost_to_go} and~\eqref{eq:zero_sum_cost_to_go_al1} for $\alpha_k = 0$ and $1$, respectively, the value of the state $x$ at time $k$ satisfies,
	\begin{equation}
		\begin{split}
			\label{eq:zero_sum_matrix_al0}
			V_{k}^{0}(x) &= g(x) + d(x) + V_{k+1}^{0}(f_k^{0}(x)) - \frac{d(x)a(x)}{\tilde{V}_{k+1}(x)} 
		\end{split},
	\end{equation}
	\begin{equation}
		\begin{split}
			\label{eq:zero_sum_matrix_al1}
			V_{k}^{1}(x) &= g(x) - a(x) + V_{k+1}^{1}(f_k^{0}(x)) + \frac{d(x)a(x)}{\tilde{V}_{k+1}(x)} 
		\end{split},
	\end{equation}
	where $\tilde{V}_{k+1}(x) := V_{k+1}^{1}(f_k^{1}(x)) - V_{k+1}^{0}(f_k^{0}(x))$. \frqed
\end{theorem}
\begin{proof}
    Given any zero-sum game matrix
	\begin{equation*}
		M = \begin{bmatrix}
			m_{1} & m_{2} \\
			m_{3} & m_{4}
		\end{bmatrix}
	\end{equation*}
	that does not admit any row or column domination, the unique mixed policy of the row and column player and the value of the game (see e.g.,~\cite{banik2020secure}) are given by
	\begin{equation}
		\label{eq:zero_sum_policy}
		\begin{split}
			\pi^{*}_\text{row} = \begin{bmatrix}
				\frac{m_{4} - m_{3}}{m_{1} - m_{2} + m_{4} - m{3}} \\[2ex]
				\frac{m_{1} - m_{2}}{m_{1} - m_{2} + m_{4} - m{3}} 
			\end{bmatrix},
			\pi^{*}_\text{col} = \begin{bmatrix}
				\frac{m_{4} - m_{2}}{m_{1} - m_{2} + m_{4} - m{3}} \\[2ex]
				\frac{m_{1} - m_{3}}{m_{1} - m_{2} + m_{4} - m{3}} 
			\end{bmatrix},
		\end{split}
	\end{equation} 
	\begin{equation}
		\label{eq:zero_sum_game}
		\text{Value of M} = {\pi^{*}}^T_\text{row}M\pi^{*}_\text{col} = \frac{m_{1}m_{4} - m_{2}m_{3}}{m_{1} - m_{2} + m_{4} - m{3}}.
	\end{equation}
	Substituting the entries of the cost-to-go matrices from~\eqref{eq:Cost_to_go} and~\eqref{eq:zero_sum_cost_to_go_al1} in~\eqref{eq:zero_sum_game}, and using the expressions in~\eqref{eq:V_k^0_cost_to_go} and~\eqref{eq:V_k^1_cost_to_go}, we obtain~\eqref{eq:zero_sum_matrix_al0} and~\eqref{eq:zero_sum_matrix_al1} for $\alpha = 0$ and $1$, respectively.
\end{proof}
\medskip

Theorem 1 states that for a finite cardinality of the state $x$ and over a finite horizon $L$, we obtain the exact value of the state and the saddle point of the \texttt{FlipDyn} game. However, the computational and storage complexity of the recursive approach will scale undesirably for continuous state spaces. For this purpose, we will provide a parametric form of the value function for the case of linear dynamics with quadratic costs in the next section.

\section{\texttt{FlipDyn} Control for LQ Problems}\label{sec:FlipDyn_Linear_Systems}
For linear dynamics and quadratic costs, we divide our analysis into two cases, a scalar system ($1$-dimensional) and an $n$-dimensional system. Under these assumptions, the \texttt{FlipDyn} setup~\eqref{eq:dyn_sys_compact} reduces to
\begin{align}\label{eq:linear_dynamics}
	x_{k+1} = F_{k}x_{k} + (1 - \alpha_{k})B_{k}u_{k} + \alpha_{k}E_{k}w_{k},
\end{align}
where $F_{k} \in \mathbb{R}^{n \times n}$ is the state transition matrix, $B_{k} \in \mathbb{R}^{n \times m}$ is the control matrix, and $E_{k} \in \mathbb{R}^{n\times p}$ is the attack matrix. 


It has been shown in~\cite{kwakernaak1972linear} that the optimal control law for any linear time system is achieved using state-feedback information (e.g.,~\cite{blanchini1990feedback,kaminer1993mixed}). Therefore, in this work, we will assume a state-feedback controller for both players. Particularly, we assume that
\begin{align}
    \label{eq:state_feedback_def_adv}
	u_k = -K_{k}x_k, \qquad w_k = W_{k}x_k,
\end{align}
where $K_{k} \in \mathbb{R}^{m\times n}, W_{k} \in\mathbb{R}^{p\times n}$ are possibly time varying matrices denoting the defender's and adversary's control gains, respectively. 
We will now simplify the recursive equations~\eqref{eq:zero_sum_matrix_al0} and~\eqref{eq:zero_sum_matrix_al1} under the following assumed costs.
\begin{assumption}[Quadratic state-dependent costs]\label{ast:N_D_sys_quad_costs}
The stage and takeover costs for each player satisfy
\begin{equation}
    \label{eq:ST_MV_cost_Q}
    \begin{aligned}
        g(x) = x^{T}Qx, \quad d(x) = x^{T}Dx, \quad a(x) = x^{T}Ax,
    \end{aligned}
\end{equation}
where $Q, D$ and $A$ are given positive definite matrices.
\end{assumption}
Under Assumption~\ref{ast:N_D_sys_quad_costs}, the recursions in~\eqref{eq:zero_sum_matrix_al0} and~\eqref{eq:zero_sum_matrix_al1} yield 
\begin{equation}
	\begin{split}
		\label{eq:V_al0_QC}
		V_{k}^{0}(x) &= x^{T}(Q + D)x +  V_{k+1}^{0}(f_k^{0}(x)) - \frac{x^{T}Dxx^{T}Ax}{\widetilde{V}_{k+1}(x)} 
	\end{split}
\end{equation}
\begin{equation}
	\begin{split}
		\label{eq:V_al1_QC}
		V_{k}^{1}(x) &= x^{T}(Q - A)x +  V_{k+1}^{1}(f_k^{1}(x)) + \frac{x^{T}Dxx^{T}Ax}{\widetilde{V}_{k+1}(x)} 
	\end{split},
\end{equation}
where $\tilde{V}_{k+1}(x)$ has been defined in the statement of Theorem~\ref{th:Cost_to_go_zerosum_value}.

Assuming a parametric form for the value function corresponding to $\alpha = 0$ and $1$ as,
\begin{align*}
    \label{eq:para_form_al_QC}
    V_{k}^{0}(x) = x^{T}P^{0}_kx, \\
    V_{k}^{1}(x) = x^{T}P^{1}_kx,
\end{align*}
where $P^{0}_{k}$ and $P^{1}_{k}$ are positive semi-definite matrices corresponding to the \texttt{FlipDyn} states $\alpha = 0$ and $1$, respectively. Therefore, the value function~\eqref{eq:V_al0_QC} and~\eqref{eq:V_al1_QC} under this parametric form satisfy
\begin{equation}
	\begin{split}
		\label{eq:P_V_al0_QC}
		V_{k}^{0}(x) = x^{T}(Q + D + \widetilde{B}_{k}^{T}P^{0}_{k+1}\widetilde{B}_{k})x -  \frac{x^{T}Dxx^{T}Ax}{x^{T}\widetilde{P}_{k+1}x},
	\end{split}
\end{equation}
\begin{equation}
	\begin{split}
		\label{eq:P_V_al1_QC}
		V_{k}^{1}(x) &= x^{T}(Q - A + \widetilde{W}_{k}^{T}P^{1}_{k+1}\widetilde{W}_{k})x +  \frac{x^{T}Dxx^{T}Ax}{x^{T}\widetilde{P}_{k+1}x} 
	\end{split},
\end{equation}
where $\widetilde{W}_{k} := (F_{k} + B_{k}W_{k}), \widetilde{B}_{k} := (F_{k} - B_{k}K_{k})$ and $\widetilde{P}_{k+1} := \widetilde{W}_{k}^{T}P^{1}_{k+1}\widetilde{W}_{k} - \widetilde{B}_{k}^{T}P^{0}_{k+1}\widetilde{B}_{k}$.
This quadratic form yields the following expressions for the mixed policies of each player at time $k$ as summarized in the following result.

\begin{cor}
\label{cor:policy_general_QC}
For the linear dynamics~\eqref{eq:linear_dynamics} and affine controls~\eqref{eq:state_feedback_def_adv}, under Assumption~\ref{ast:N_D_sys_quad_costs} the players' policies satisfy
\begin{align}
		\label{eq:NE_policy_alpha0_al0_VC_QC}
		\begin{split}
			y^{*}_{k| \alpha_k = 0}(x)  = \begin{bmatrix} \hat{\beta}_{k}^{*}(x) & 1 - \hat{\beta}_{k}^{*}(x)
			\end{bmatrix}^{T},
		\end{split}
	\end{align}
	\begin{align}
		\label{eq:NE_policy_alpha1_al0_VC_QC}
		\begin{split}
			z^{*}_{k | \alpha_k = 0}(x)  = \begin{bmatrix} \hat{\gamma}_{k}^{*}(x) & 1 - \hat{\gamma}_{k}^{*}(x)
			\end{bmatrix}^{T},
		\end{split}
	\end{align}
    \begin{align}
		\label{eq:NE_policy_alpha0_al1_VC_QC}
		\begin{split}
			y^{*}_{k| \alpha_k = 1}(x)  = \begin{bmatrix} 1 - \hat{\beta}_{k}^{*}(x) & \hat{\beta}_{k}^{*}(x)
			\end{bmatrix}^{T},
		\end{split}
	\end{align}
	\begin{align}
		\label{eq:NE_policy_alpha1_al1_VC_QC}
		\begin{split}
			z^{*}_{k | \alpha_k = 1}(x)  = \begin{bmatrix} 1 - \hat{\gamma}_{k}^{*}(x) & \hat{\gamma}_{k}^{*}(x) 
			\end{bmatrix}^{T},
		\end{split}
	\end{align}
	where,
	\begin{align*}
	    \hat{\beta}_{k}^{*} = \frac{x^{T}Ax}{x^{T}\widetilde{P}_{k+1}x}, \quad  1 - \hat{\gamma}_{k}^{*} = \frac{x^{T}Dx}{x^{T}\widetilde{P}_{k+1}x}.
	\end{align*}
The terms $y^{*}_{k|\alpha_k}$ and $z^{*}_{k|\alpha_k}$ correspond to the defender's and adversary's policy for the \texttt{FlipDyn} state $\alpha_{k}$ at time $k$, respectively. \frqed
\end{cor}
 Substituting $\beta_{k}^{*}$ from~\eqref{eq:NE_policy_alpha0_al0_VC_QC} in~\eqref{eq:P_V_al0_QC}, and $1 - \gamma_{k}^{*}$ from~\eqref{eq:NE_policy_alpha1_al1_VC_QC} in~\eqref{eq:P_V_al1_QC},  we obtain the following form, 
\begin{equation}
	\begin{split}
		\label{eq:P_V_al0_QC_pol_ver}
		V_{k}^{0}(x) = x^{T}(Q + D + \widetilde{B}^{T}P^{0}_{k+1}\widetilde{B})x - x^{T}Dx (\hat{\beta}^{*}_{k}(x)),
	\end{split}
\end{equation}
\begin{equation}
	\begin{split}
		\label{eq:P_V_al1_QC_pol_ver}
		V_{k}^{1}(x) = x^{T}(Q - A + \widetilde{W}^{T}P^{1}_{k+1}\widetilde{W})x + x^{T}Ax (1-\hat{\gamma}^{*}_{k}(x)).
	\end{split}
\end{equation}
 We observe that both~\eqref{eq:P_V_al0_QC_pol_ver} and~\eqref{eq:P_V_al1_QC_pol_ver} are nonlinear in $x$. Therefore, a quadratic parameterization cannot necessarily represent the value function with quadratic costs. However, we show that for a scalar system ($1$-dimensional), this parameterization is sufficient. This parameterization also yields an approximation for the value functions in $n$-dimensional state spaces as detailed below.

\medskip

\subsubsection{\bf Scalar/1-dimensional system}
The state, defense and attack costs for a scalar system simplify to
\begin{equation}
\label{eq:cost_form_1D}
    g(x) = g x^{2}, \quad d(x) = d x^{2}, \quad a(x) = a x^{2},
\end{equation}
where $g, d$ and $a$ are positive constants and $x \in \mathbb{R}$. The following result provides a closed-form expression for the NE of the \texttt{FlipDyn} game and the corresponding value of the state at time instant $k$.

\begin{theorem}
	\label{th:Value_of_game_scalar}
	The unique mixed Nash equilibrium at any time $k$ for the \texttt{FlipDyn}  state of $\alpha_k = 0$ for a scalar system with costs~\eqref{eq:cost_form_1D} and dynamics~\eqref{eq:linear_dynamics} is given by,
	
	\begin{align}
		\label{eq:NE_policy_alpha0_SC}
		\begin{split}
				y^{*}_{k| \alpha_k = 0}  = \begin{bmatrix} \dfrac{a}{\tilde{\mathbf{p}}_{k+1}} & \dfrac{\tilde{\mathbf{p}}_{k+1} - a}{\tilde{\mathbf{p}}_{k+1}}
				\end{bmatrix}^{T},
		\end{split}
	\end{align}
	\begin{align}
		\label{eq:NE_policy_alpha1_SC}
		\begin{split}
				 z^{*}_{k| \alpha_k = 0}  = \begin{bmatrix} \dfrac{\tilde{\mathbf{p}}_{k+1} - d}{\tilde{\mathbf{p}}_{k+1}} & \dfrac{d}{\tilde{\mathbf{p}}_{k+1}}
				\end{bmatrix}^{T}
		\end{split}.
	\end{align}
	The saddle-point value at time instant $k$ is parameterized by, 
	\begin{align}
		\begin{split}
			\mathbf{p}_{k}^{0} = g + (F_{k} - B_{k}K_{k})^{2}\mathbf{p}^{0}_{k+1} + d -\frac{da}{\tilde{\mathbf{p}}_{k+1}}
		\end{split},
		\label{eq:NE_Val_SC}
	\end{align}
	where $\tilde{\mathbf{p}}_{k+1} := (F_{k} + B_{k}W_{k})^{2}\mathbf{p}^{1}_{k+1} - (F_{k} - B_{k}K_{k})^{2}\mathbf{p}^{0}_{k+1}$.
	
	Similarly, for the \texttt{FlipDyn} state of $\alpha_k = 1$, the unique Nash equilibrium at time $k$ is, 
	\begin{align}
		\label{eq:NE_policy_alpha0_al1_SC}
		\begin{split}
			y^{*}_{k|\alpha_k = 1} = \begin{bmatrix} \dfrac{\tilde{\mathbf{p}}_{k+1} - a}{\tilde{\mathbf{p}}_{k+1}} & \dfrac{a}{\tilde{\mathbf{p}}_{k+1}}  
			\end{bmatrix}^{T},
		\end{split}
	\end{align}
	\begin{align}
		\label{eq:NE_policy_alpha1_al1_SC}
		\begin{split}
			z^{*}_{k | \alpha_k = 1}  = \begin{bmatrix} \dfrac{d}{\tilde{\mathbf{p}}_{k+1}} & \dfrac{\tilde{\mathbf{p}}_{k+1} - d}{\tilde{\mathbf{p}}_{k+1}}  
			\end{bmatrix}^{T}.
		\end{split}
	\end{align}
	The saddle-point value at time $k$ is parameterized by, 
	\begin{align}
		\begin{split}
			\mathbf{p}_{k}^{1} = g + (F_{k} + B_{k}W_{k})^{2}\mathbf{p}^{1}_{k+1} - a +\frac{da}{\tilde{\mathbf{p}}_{k+1}},
		\end{split}
		\label{eq:NE_Val_SC_al1}
	\end{align}
	such that (recursively) $\mathbf{p}^{0}_k \geq 0 $ and $(F_{k} + B_{k}W_{k})^{2}\mathbf{p}^{1}_{k+1} \geq (F_{k} - B_{k}K_{k})^{2}\mathbf{p}^{0}_{k+1} + \max\{d,a\}$, hold $\forall k \in \mathbb{N}$. \frqed
\end{theorem}
\begin{proof}
	Suppose that at the terminal time instant $L$, 
	\[V_{L}^{0} = \mathbf{p}_{L}^{0}x^{2} := gx^{2}.\]
	Translating the condition~\eqref{eq:Lemma1} for mixed policy equilibrium gives the terminal condition,
	\begin{equation}
		\label{eq:V_SC_terminal}
		\mathbf{p}_{L}^{0} := g, \quad \mathbf{p}_{L}^{1} \geq g + \max\{a, d\} + \mu 
	\end{equation}
	where $\mu \geq 0$ is a constant. Substituting the costs in~\eqref{eq:cost_form_1D} into Corollary~\ref{cor:policy_general_QC}, we obtain the following policy for the defender (row player) and adversary (column player) and the value of the game, 
	\begin{align}
		\label{eq:Proof_policy_alpha0_SC}
		\begin{split}
			y^{*}_{k}(x) = \begin{bmatrix} \dfrac{ax^2}{\tilde{\mathbf{p}}_{k+1}x^2} & \dfrac{(\tilde{\mathbf{p}}_{k+1} - a)x^2}{\tilde{\mathbf{p}}_{k+1}x^2}
			\end{bmatrix},
		\end{split}
	\end{align}
	\begin{align}
		\label{eq:Proof_policy_alpha1_SC}
		\begin{split}
			z^{*}_{k}(x) = \begin{bmatrix} \dfrac{(\tilde{\mathbf{p}}_{k+1} - d)x^{2}}{\tilde{\mathbf{p}}_{k+1}x^{2}} & \dfrac{dx^{2}}{\tilde{\mathbf{p}}_{k+1}x^{2}}
			\end{bmatrix},
		\end{split}
	\end{align}
	\begin{align}
		\begin{split}
			\mathbf{p}_{k}^{0}x^{2} &= \left(g + (F_{k} - B_{k}K_{k})^{2}\mathbf{p}^{0}_{k+1} + d\right)x^{2} -\frac{dx^{2}ax^{2}}{\tilde{\mathbf{p}}_{k+1}x^{2}}
		\end{split},
		\label{eq:Proof_Val_SC}
	\end{align}
	where $\tilde{\mathbf{p}}_{k+1} := (F_{k} + B_{k}W_{k})^{2}\mathbf{p}^{1}_{k+1} - (F_{k} - B_{k}K_{k})^{2}\mathbf{p}^{0}_{k+1}$. Eliminating the state from~\eqref{eq:Proof_policy_alpha0_SC},~\eqref{eq:Proof_policy_alpha1_SC} and~\eqref{eq:Proof_Val_SC}, we obtain ~\eqref{eq:NE_policy_alpha0_SC},~\eqref{eq:NE_policy_alpha1_SC} and the recursion~\eqref{eq:NE_Val_SC}, respectively. We derive similar expressions for $\alpha_k = 1$ using the same procedure by replacing $\Xi_k^0$ by $\Xi_k^1$. 
\end{proof}
\medskip

Observe that the policy for the \texttt{FlipDyn} state $\alpha = 1$ is complementary to the policy corresponding to $\alpha = 0$ indicating the need to compute the policy for any one of the \texttt{FlipDyn} states. Using Theorem~\ref{th:Value_of_game_scalar}, the saddle points of the \texttt{FlipDyn} game for $\alpha = 0$ and $1$ are,
\begin{equation}
	J_E(x, 0, y_{\mathbf{L}}^{*}, z_{\mathbf{L}}^{*}) = x^{\mathrm{T}}\mathbf{p}^{0}_0x, \quad J_E(x, 1,y_{\mathbf{L}}^{*}, z_{\mathbf{L}}^{*}) = x^{\mathrm{T}}\mathbf{p}^{1}_0x.
\end{equation}
Thus, we have obtained an exact solution for the $1$-dimensional system with the parameterized value function and the player policies for both the \texttt{FlipDyn} states. Next, we extend this approach to derive an approximate solution for an $n$-dimensional system.

\medskip

\subsubsection{\bf $n$-dimensional system}
To address the nonlinearity of the value function, we first introduce an approximation that will enable recursive computation of the parameters defining the value function, thus making it independent of the state. 
\begin{theorem}
\label{cor:Val_QC_ND}
    At any time instant $k \in \mathbb{N}$, under Assumption 1, suppose that the nonlinear terms $x^{T}Dx \, \hat{\beta}^{*}_{k}(x)$ and $x^{T}Ax \, (1-\hat{\gamma}^{*}_{k}(x))$ 
    in~\eqref{eq:P_V_al0_QC_pol_ver} and in~\eqref{eq:P_V_al1_QC_pol_ver} can be upper bounded by a common quadratic form in the state, i.e.,
    \begin{align}\label{eq:upperbound1}
        (x^{T}Dx)\hat{\beta}^{*}_{k}(x) & \leq x^{T}D(\widetilde{P}_{k+1})^{-1}Ax, \\
        (x^{T}Ax)(1-\hat{\gamma}^{*}_{k}(x)) & \leq x^{T}D(\widetilde{P}_{k+1})^{-1}Ax, \label{eq:upperbound2}
    \end{align}
    where $\widetilde{P}_{k+1} := \widetilde{W}_{k}^{T}P^{1}_{k+1}\widetilde{W}_{k} - \widetilde{B}_{k}^{T}P^{0}_{k+1}\widetilde{B}_{k}, \ \widetilde{W}_{k} := (F_{k} + B_{k}W_{k})$ and $\widetilde{B}_{k} := (F_{k} - B_{k}K_{k})$. 

Then, the value functions corresponding to each \texttt{FlipDyn} state are given by $V_{k}^0(x) = x^TP_k^0x, \ V_k^1(x) = x^TP_k^1x$, where the matrices $P_k^0$ and $P_k^1$ are chosen to satisfy 
\begin{align*}
		P_{k}^{0} & \preceq Q + D + \widetilde{B}_{k}^{T}P^{0}_{k+1}\widetilde{B}_{k} - D\widetilde{P}_{k+1}^{-1}A, \\
		P_{k}^{1} & \succeq Q - A + \widetilde{W}_{k}^{T}P^{1}_{k+1}\widetilde{W}_{k} + D\widetilde{P}_{k+1}^{-1}A. 
\end{align*}
\end{theorem} 

\begin{proof} To prove this claim, we substitute \eqref{eq:upperbound1} and \eqref{eq:upperbound2} into \eqref{eq:P_V_al0_QC_pol_ver} and \eqref{eq:P_V_al1_QC_pol_ver}, respectively, to obtain
\begin{align*}
		x^{T}P_{k}^{0}x & \leq x^{T}(Q + D + \widetilde{B}_{k}^{T}P^{0}_{k+1}\widetilde{B}_{k} - D\widetilde{P}_{k+1}^{-1}A)x, \\
		x^{T}P_{k}^{1}x & \geq x^{T}(Q - A + \widetilde{W}_{k}^{T}P^{1}_{k+1}\widetilde{W}_{k} + D\widetilde{P}_{k+1}^{-1}A)x.
\end{align*}
Since these inequalities hold for every $x \in \mathbb{R}^n$, the above claim is proven. 
\end{proof}

\medskip

The next result shows that conditions \eqref{eq:upperbound1} and \eqref{eq:upperbound2} do hold for a special class of matrices $A$ and $D$. 
\begin{proposition}
Conditions \eqref{eq:upperbound1} and \eqref{eq:upperbound2} hold for any positive definite matrix $\widetilde{P}_{k+1}$ if 
\[
A = aI, \text{ and } D = dI, \quad \text{ for any $a, d > 0$}.
\]
\end{proposition}
\begin{proof}
Substituting these choices of $A$ and $D$ into \eqref{eq:upperbound1} and \eqref{eq:upperbound2} yield the following inequality to be established.
\begin{equation}\label{eq:prop1}
\frac{x^Tx}{x^T\widetilde{P}_{k+1}x} \leq \frac{x^T\widetilde{P}_{k+1}^{-1}x}{x^Tx}.
\end{equation}
Setting, 
\[
\Gamma := \widetilde{P}_{k+1}^{1/2}x, \Phi := \widetilde{P}_{k+1}^{-1/2}x,
\]
observe that the $2\times 2$ matrix
\[
\mathcal{M} := \begin{bmatrix} \Gamma^T\Gamma & \Gamma^T\Phi \\ \Phi^T\Gamma & \Phi^T\Phi \end{bmatrix} = \begin{bmatrix} \Gamma & \Phi \end{bmatrix}^T\begin{bmatrix} \Gamma & \Phi \end{bmatrix} \succeq 0.
\]
Therefore, 
\[
\text{det}(\mathcal{M}) = (\Gamma^T\Gamma)(\Phi^T\Phi) - (\Gamma^T\Phi)^2 \geq 0,
\]
and thus our claim \eqref{eq:prop1} holds.
\end{proof}

\medskip

Theorem~\ref{cor:Val_QC_ND} enables a recursive computation for an approximate value function  independently of the state using the parameters,
\begin{align}
	\label{eq:Val_al0_QC_approx_final}
		\hat{P}_{k}^{0} &= Q + D + \widetilde{B}_{k}^{T}\hat{P}^{0}_{k+1}\widetilde{B}_{k} - D\check{P}_{k+1}^{-1}A, \\
	\label{eq:Val_al1_QC_approx_final}
		\hat{P}_{k}^{1} &= Q - A + \widetilde{W}_{k}^{T}\hat{P}^{1}_{k+1}\widetilde{W}_{k} + D\check{P}_{k+1}^{-1}A,
\end{align}
where  $\check{P}_{k+1} := \widetilde{W}_{k}^{T}\hat{P}^{1}_{k+1}\widetilde{W}_{k} - \widetilde{B}_{k}^{T}\hat{P}^{0}_{k+1}\widetilde{B}_{k}, \ \widetilde{W}_{k} := (F_{k} + B_{k}W_{k})$ and $\widetilde{B}_{k} := (F_{k} - B_{k}K_{k})$
such that 
    $$\check{P}_{k+1} \succcurlyeq A \text{ and } \check{P}_{k+1} \succcurlyeq D, \quad \forall k \in \{1,2,\dots,L\}.$$
We initialize the parameterized value function at the terminal time instant $L$ as,
\begin{equation}
    \label{eq:V_N_terminal}
    \hat{P}_{L}^{0} = Q, \quad \hat{P}_{L}^{1} = \begin{cases}
        Q + A + \mu I, & \text{if } A \succcurlyeq D \\
        Q + D + \mu I, & \text{otherwise }
    \end{cases} ,
\end{equation}
where $\mu$ is a  constant. We demonstrate numerical results of this procedure in the next section.


\begin{figure*}[h]
	\begin{center}
		\subfloat[]{\includegraphics[width = 0.24\linewidth]{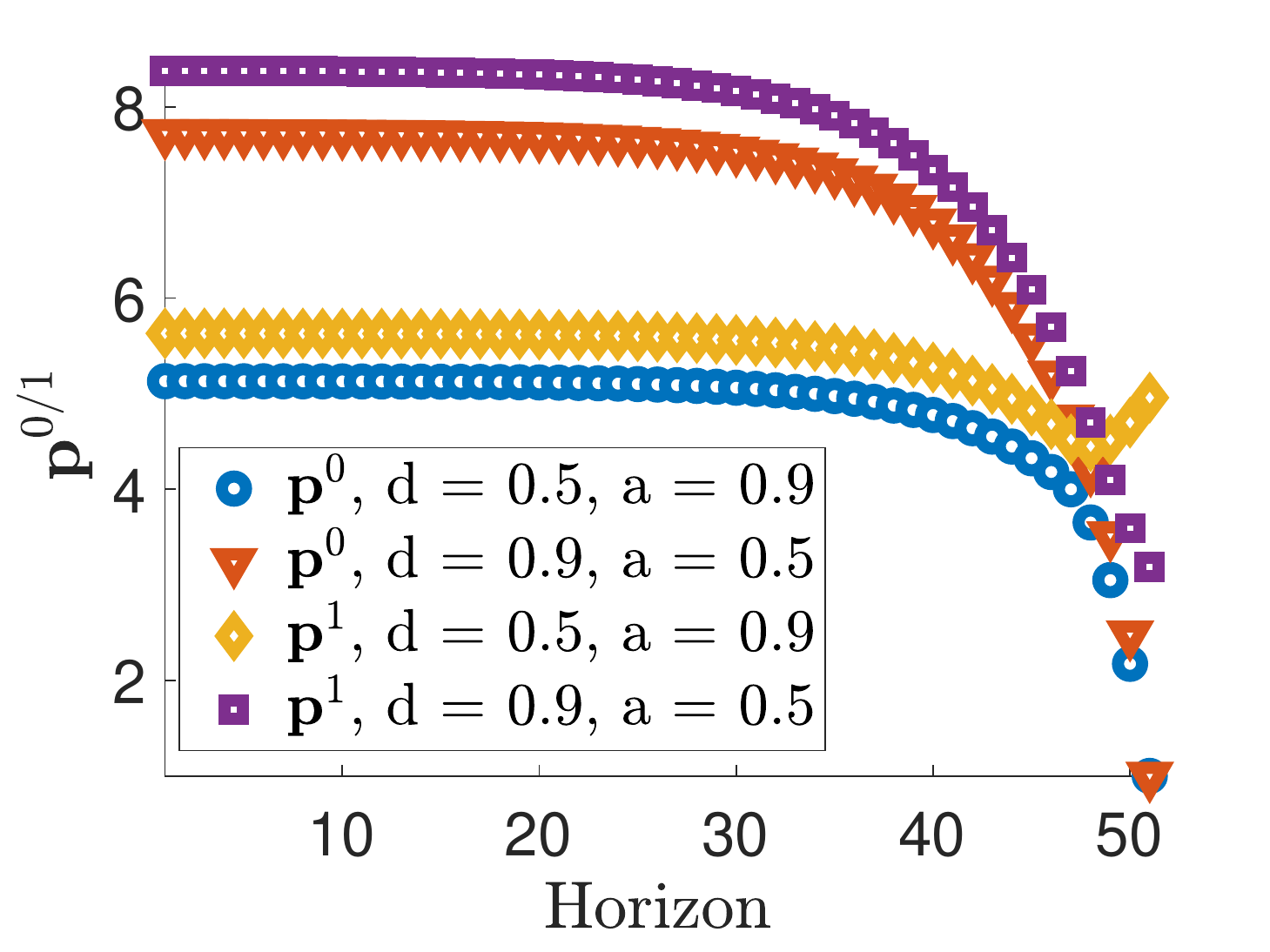}
			\label{fig:p01_1D}	
		}
		\subfloat[]{\includegraphics[width = 0.24\linewidth]{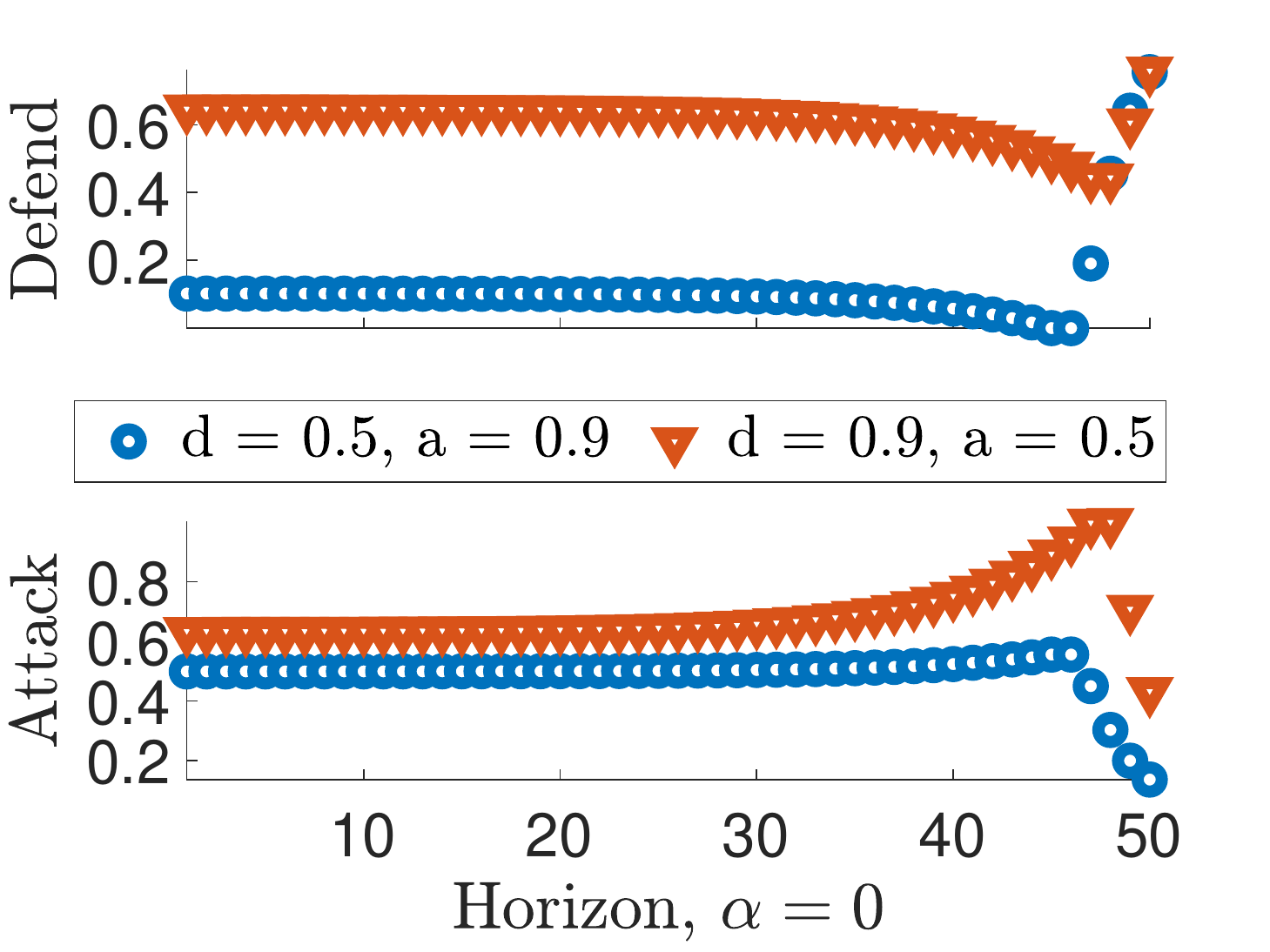}
			\label{fig:1D_policy}	
		}
		\subfloat[]{\includegraphics[width = 0.24\linewidth]{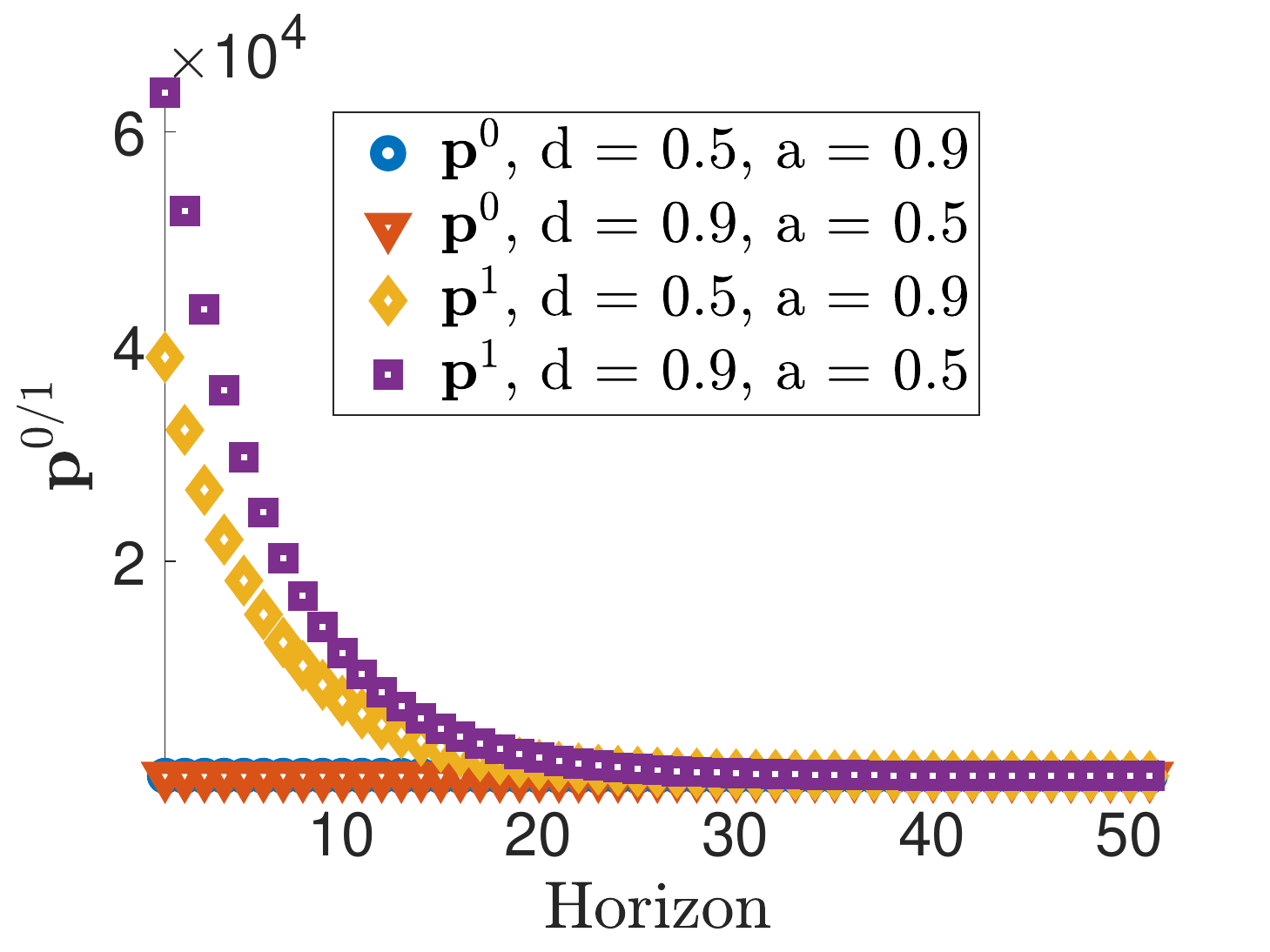}
			\label{fig:p01_1D_UB}	
		}
		\subfloat[]{\includegraphics[width = 0.24\linewidth]{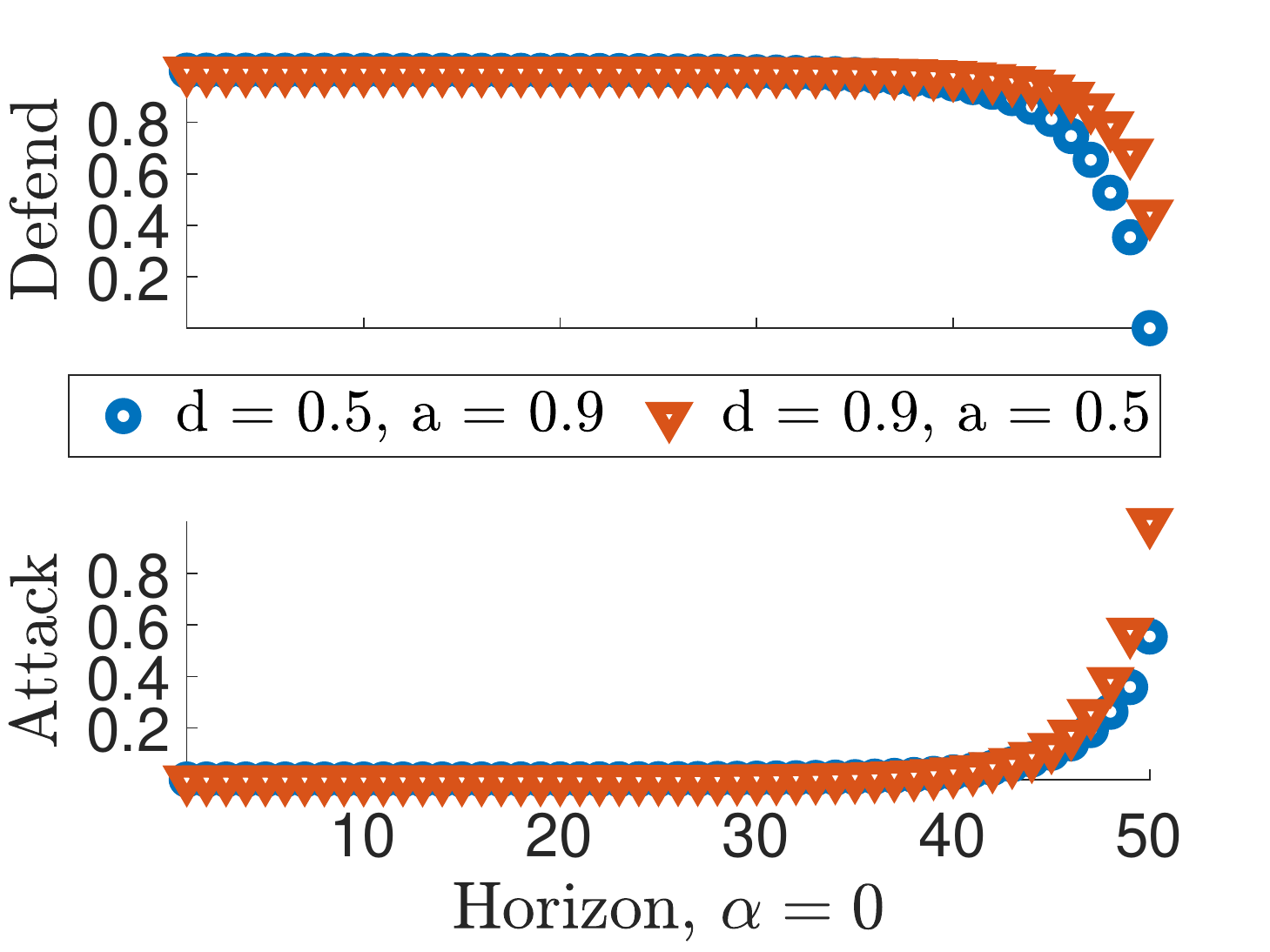}
			\label{fig:1D_policy_UB}	
		}
		\caption{\small (a) Coefficient of the parameterized value function, $\mathbf{p^{0}}$ and $\mathbf{p^{1}}$ for a 1-dimensional system where the state is bounded ($F \leq 1$) over a horizon length of $L = 50$. (b) Attack and defense policy corresponding to the value function in Figure~\ref{fig:p01_1D} for the given set of costs. (c) Coefficient of the parameterized value function, $\mathbf{p^{0}}$ and $\mathbf{p^{1}}$ for an unbounded ($F \geq 1$) 1-dimensional system with a horizon length of $L = 50$. (d) Policy of defense and attack for the obtained parameterized value function indicated in Figure~\ref{fig:p01_1D_UB}.}
		\label{fig:FlipDyn_scalar_case}
	\end{center}
\end{figure*}
\begin{figure*}[h]
	\begin{center}
		\subfloat[]{\includegraphics[width = 0.24\linewidth]{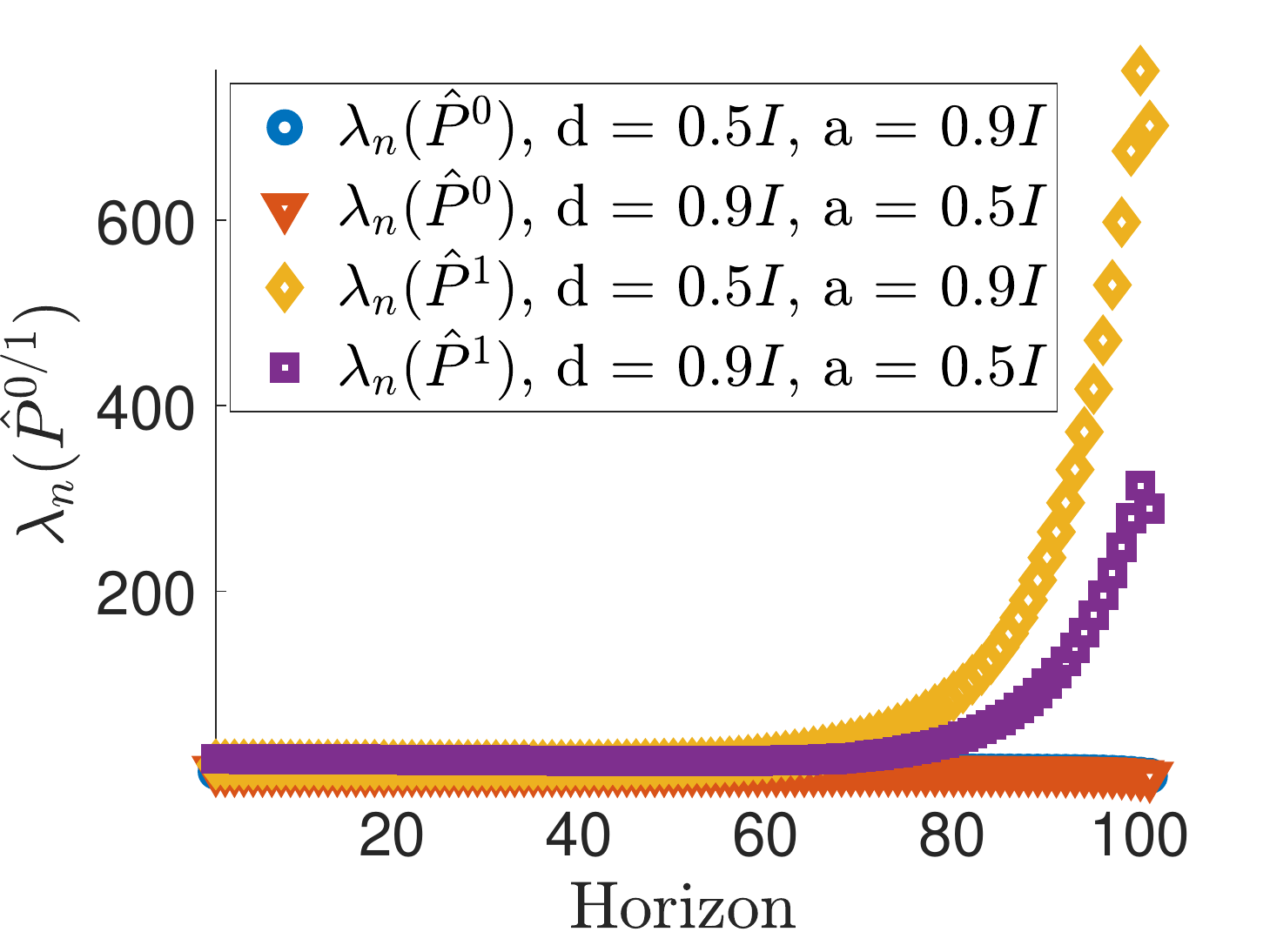}
			\label{fig:P01_ND}	
		}
		\subfloat[]{\includegraphics[width = 0.24\linewidth]{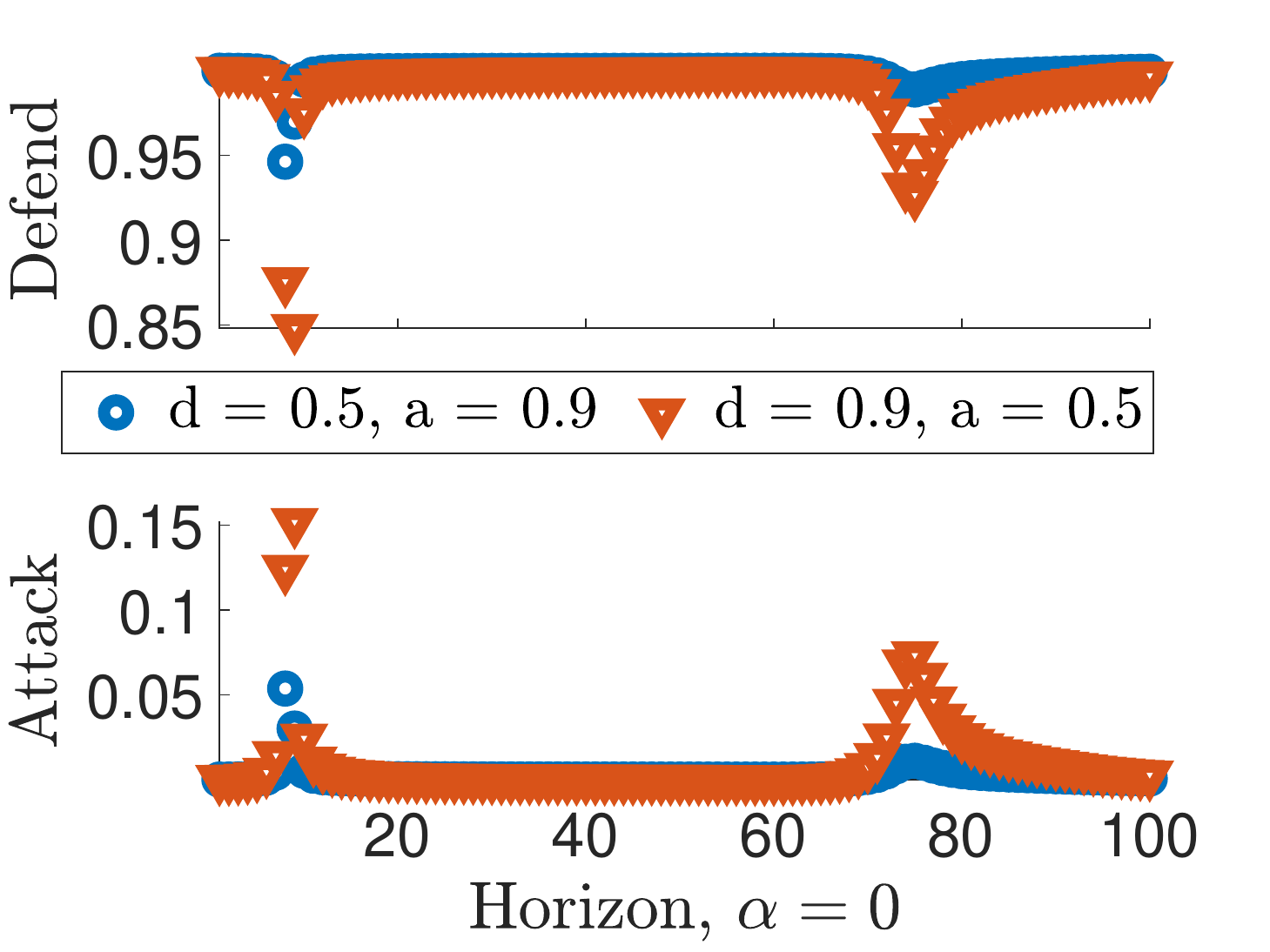}
			\label{fig:ND_policy}	
		}
		\subfloat[]{\includegraphics[width = 0.24\linewidth]{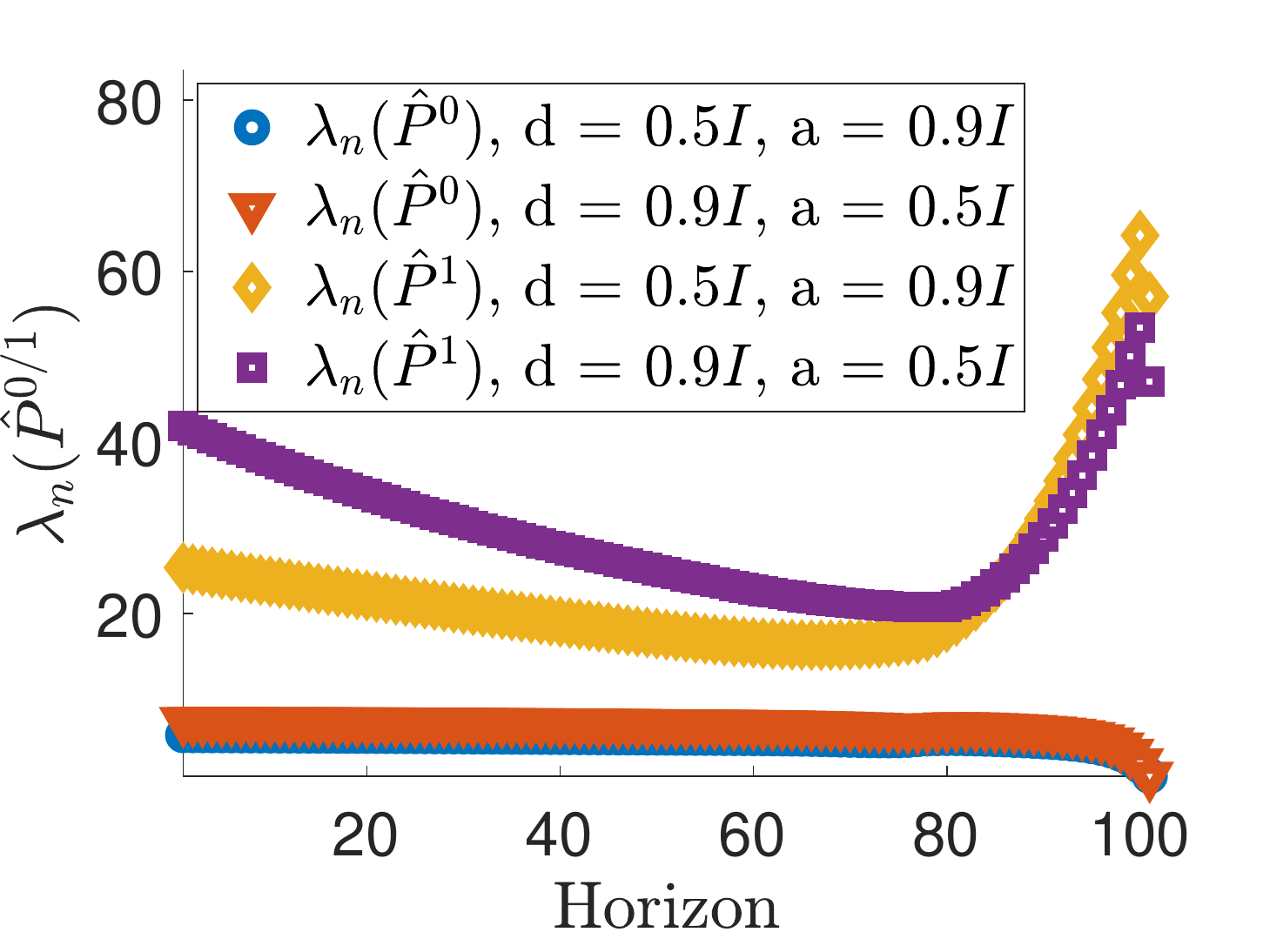}
			\label{fig:P01_ND_UB}	
		}
		\subfloat[]{\includegraphics[width = 0.24\linewidth]{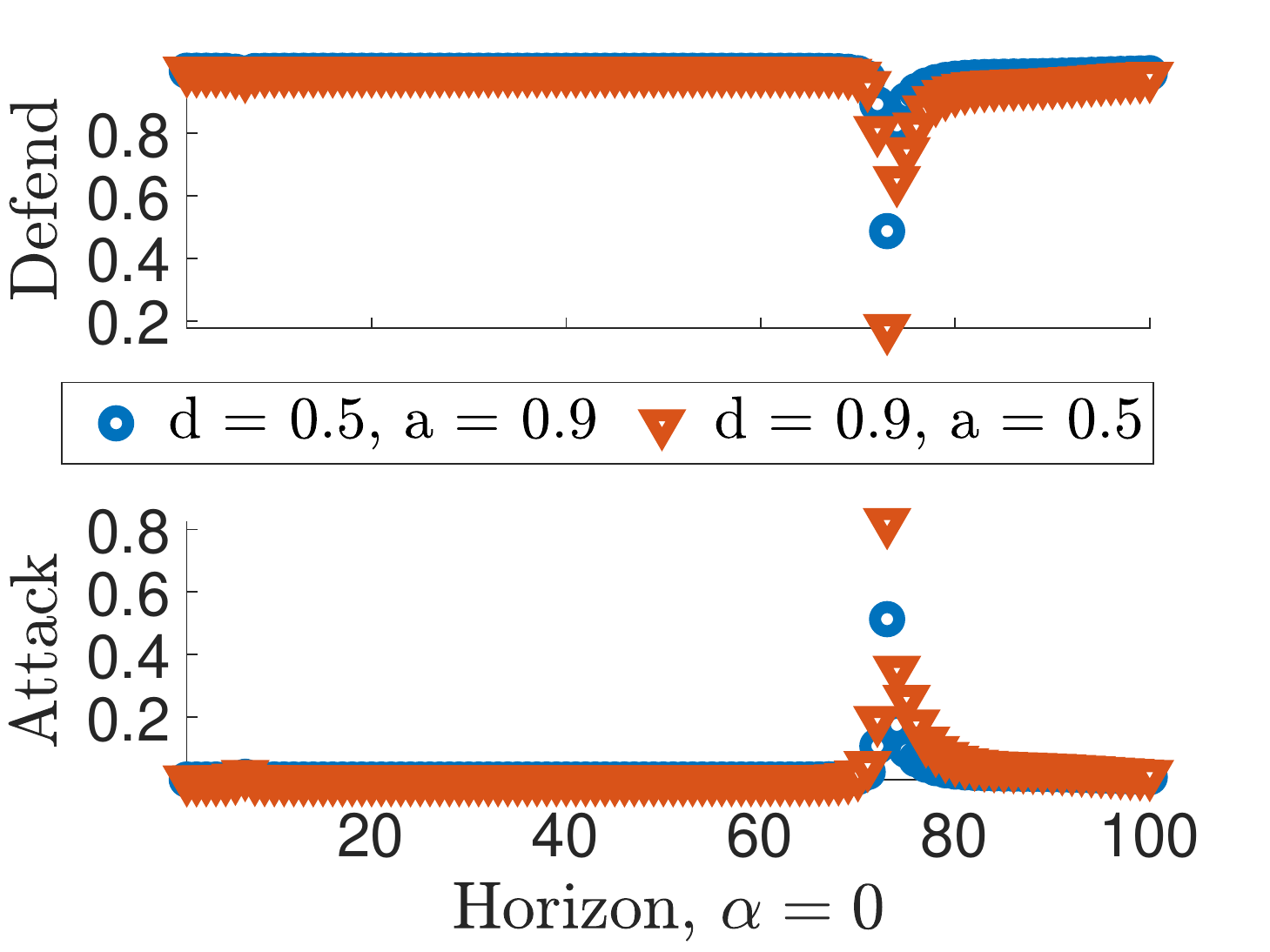}
			\label{fig:ND_policy_UB}	
		}
		\caption{\small (a) Minimum eigenvalue of the parameterized semi-definite value function matrices, $\lambda_{n}(\hat{P}^{0})$ and $\lambda_{n}(\hat{P}^{1})$ for a $n$-dimensional system, where the state is bounded ($\hat{f} \leq 1$) over a horizon length of $L = 100$. (b) Attack and defense policies corresponding to the value function in Figure~\ref{fig:P01_ND} for the given set of state, defense and attack costs. (c)  Minimum eigenvalue of the parameterized semi-definite value function matrices, $\lambda_{n}(\hat{P}^{0})$ and $\lambda_{n}(\hat{P}^{1})$ for an unbounded ($\hat{f} \geq 1$) $n$-dimensional system for the same horizon length of $L = 100$. (d) Attack and defense policies for the corresponding value function from Figure~\ref{fig:P01_ND_UB} given the \texttt{FlipDyn} state $\alpha = 0$.}	
		\label{fig:FlipDyn_vector_case}
	\end{center}
	\vspace{-0.25in}
\end{figure*}
\section{Numerical Evaluation}\label{sec:Numerics}
In this section, we will evaluate our analytic results from Theorem~\ref{th:Value_of_game_scalar} and Theorem~\ref{cor:Val_QC_ND} on a linear time-invariant system (LTI) in conjunction with a linear quadratic regulator (LQR) control law used by the defender. Without any loss of generality, we make the following mild assumption.
\begin{assumption}
	\label{ast:Intial_control}
	The system is under the defender's control at time $k = 0$, i.e, $\alpha_{0} := 0$. 
\end{assumption}	
Assumption~\ref{ast:Intial_control} is only for convenience and is reasonable to expect that the system designer would have complete control of the system upon initialization. 

Given the system description, we will now specify the parameters of the \texttt{FlipDyn} game. The dynamical system is assumed to be given by
\begin{align*}
	f^0_k(x_k) &= (F - BK)x_k, \\
	f^1_k(x_k) &= (F + BW)x_k = Fx_k, 
\end{align*}
where we have assumed that the adversary's control gain $W_k = W = \mathbf{0}, \forall k \in \{1,2,\dots, L\}$, i.e., the adversary applies zero control input commands deterring or deviating the state from reaching its equilibrium state. We now report the results of simulating the policies, and value of the \texttt{FlipDyn} game for a $1$-dimensional (1-D) and an $n$-dimensional ($n$-D) system below.
\subsubsection{1-Dimensional system}
In this setting, we use a single integrator system of the form 
\begin{align*}
	f^0_k(x_k) &= (F - \Delta K)x_k, \\
	f^1_k(x_k) &= (F + BW)x_k = Fx_k, 
\end{align*}
where $\Delta$ is the sample time. We consider two cases: $F \leq 1$ and $F > 1$. We obtain the defender gain $K$ by solving the LQR problem with arbitrarily weighted state and control cost. We solve for the value function coefficients ($\mathbf{p}_{k}^{0},\mathbf{p}_{k}^{1}$) and its corresponding policies over a time horizon of $L = 50$ for $F = 0.99$ in Figure~\ref{fig:p01_1D}, and for $F = 1.1$ in Figure~\ref{fig:p01_1D_UB}. 

We observe that the coefficients of the value function are bounded and reach an asymptotic value for $F = 0.99$, i.e., when the state is bounded so are the value coefficients, whereas for $F = 1.1$, the coefficient $\mathbf{p}^{1}$ keeps growing and is unbounded, thus indicating a very large incentive for an adversary in the initial stages of the \texttt{FlipDyn} game. The attack and defense policy conditioned on the \texttt{FlipDyn} state $\alpha = 0$, for bounded and unbounded value function coefficients are shown in Figure~\ref{fig:1D_policy} and~\ref{fig:1D_policy_UB}, respectively. For the case when $F = 0.99$, we observe that the defense and attack policy are initially dynamic and gradually converge to a stationary policy. In the case of $F > 1$, as the horizon $L$ keeps increasing, both the defender and adversary policies converge to a pair of pure policies of always defending and not attack when the \texttt{FlipDyn} state is $\alpha = 0$, reflective of the unbounded value coefficients. 

\subsubsection{$n$-Dimensional system}
We use a double integrator dynamics $(n = 2)$ of the form 
\begin{align*}
	f^0_k(x_k) &= \Bigg(\underbrace{\begin{bmatrix}
		    \hat{f} & \Delta \\ 0 & \hat{f}
		\end{bmatrix}}_{F} -  \underbrace{\begin{bmatrix}
		    0.5\Delta^{2} \\ \Delta
		\end{bmatrix}}_{B}K\Bigg)x_{k}, \\
	f^1_k(x_k) &= Fx_k, 
\end{align*}
where $\Delta > 0$ is the sample time analogous to the scalar case. The system represents a second order system with acceleration as the control input. Analogous to the scalar case, we obtain the defender's gain $K$ using the LQR method. We solve the approximate parameterized value function matrices ($\hat{P}^{0}_{k}, \hat{P}^{1}_{k}, \forall k \in \{1,\dots,L\}$) for both  \texttt{FlipDyn} states over a horizon length $L = 100$. The minimum eigenvalue of the value function matrices are shown in Figures~\ref{fig:P01_ND} and~\ref{fig:P01_ND_UB}, corresponding to $\hat{f} := 0.99$ and $\hat{f} := 1.01$, respectively. 

Akin to the scalar case, we observe a similar trend of converging coefficients when $\hat{f} \leq 1$, i.e., the system remains bounded upon lack of control, whereas the coefficients diverge away for $\hat{f} > 1$ indicating a large incentive for an adversary in the initial time instants of the \texttt{FlipDyn} game. Since the player policies for the $n$-dimensional case are functions of state, and the \texttt{FlipDyn} state is a random variable, the attack and defense policies averaged over 500 independent simulations for $\hat{f} := 0.99$ and $\hat{f} := 1.01$ are shown in Figures~\ref{fig:ND_policy} and~\ref{fig:ND_policy_UB}, respectively, with the initial state $x_{0} = \begin{bmatrix}
    0 & 1
\end{bmatrix}^{T}$. We observe a dynamic policy over the horizon length for the case of $\hat{f} := 0.99$, and a converging pure policy for $\hat{f} := 1.01$ for the \texttt{FlipDyn} state $\alpha = 0$, respectively. The converging pure policy for $\hat{f} := 1.01$ is reflective of the ever increasing value of the adversary over the horizon length. 
\medskip

\subsubsection{Regaining control}
For the above $n$-dimensional system, we observe that the \texttt{FlipDyn} state evolution over the horizon length in Figure~\ref{fig:FD_normal} for $\hat{f} \leq 1$ and $\hat{f} > 1$, respectively. The \texttt{FlipDyn} state evolution for the dynamical  system is averaged over 500 independent simulation runs. We carry out simulations to observe recovery of control by the defender, i.e., for the described $n$-D system after intentionally switching the control over to an adversary. From Figure~\ref{fig:FD_recover}, we conclude that the system moves towards recovery, i.e., back to the defender's control indicated by the \texttt{FlipDyn} state moving towards zero over the horizon length. Furthermore, when the adversary's cost is lower than the defender's, the recovery rate is higher indicating a more aggressive defender to regain back control.

\begin{figure}[h]
	\begin{center}
		\subfloat[]{\includegraphics[width = 0.49\linewidth]{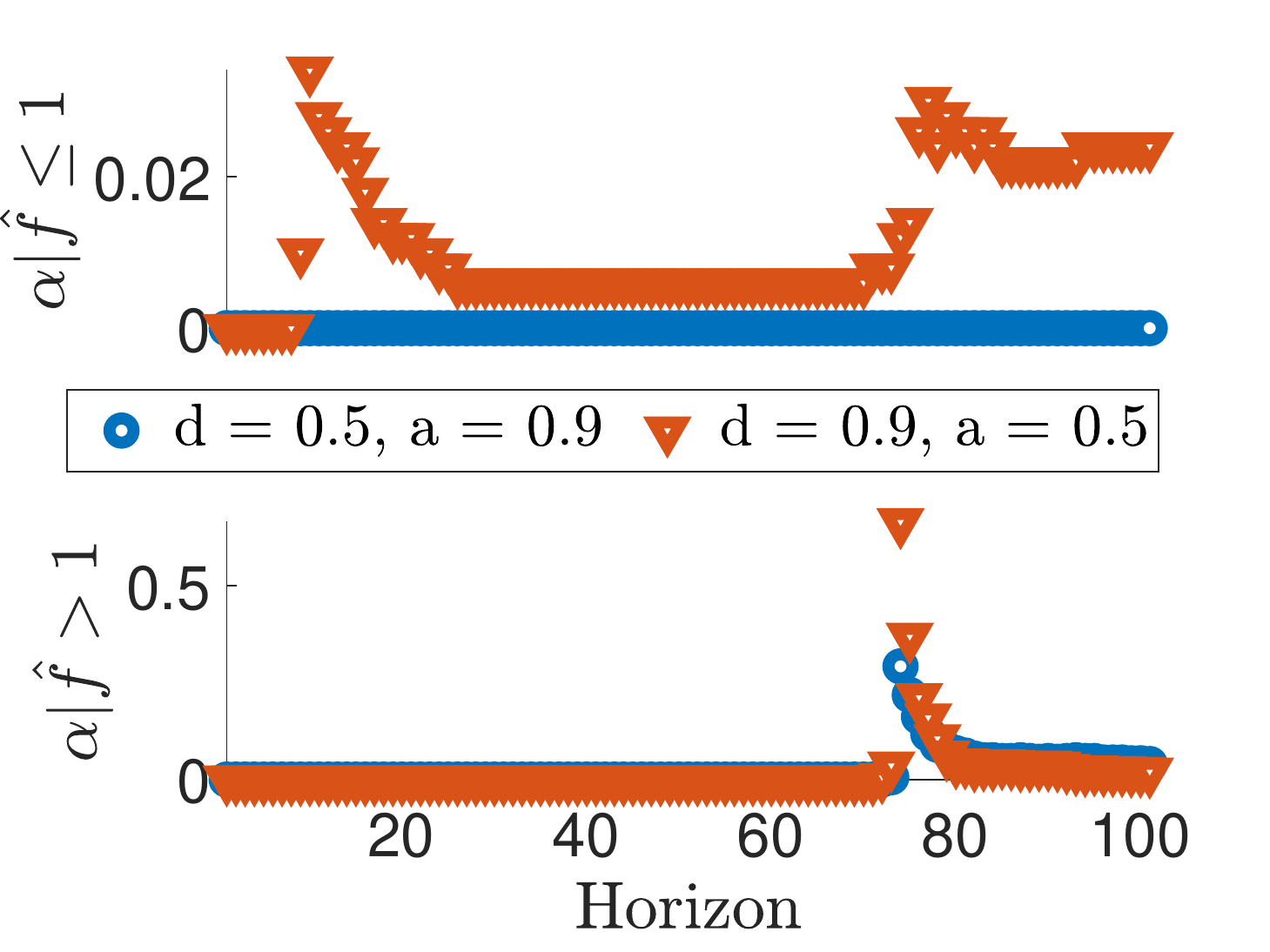}
			\label{fig:FD_normal}	
		}
		\subfloat[]{\includegraphics[width = 0.49\linewidth]{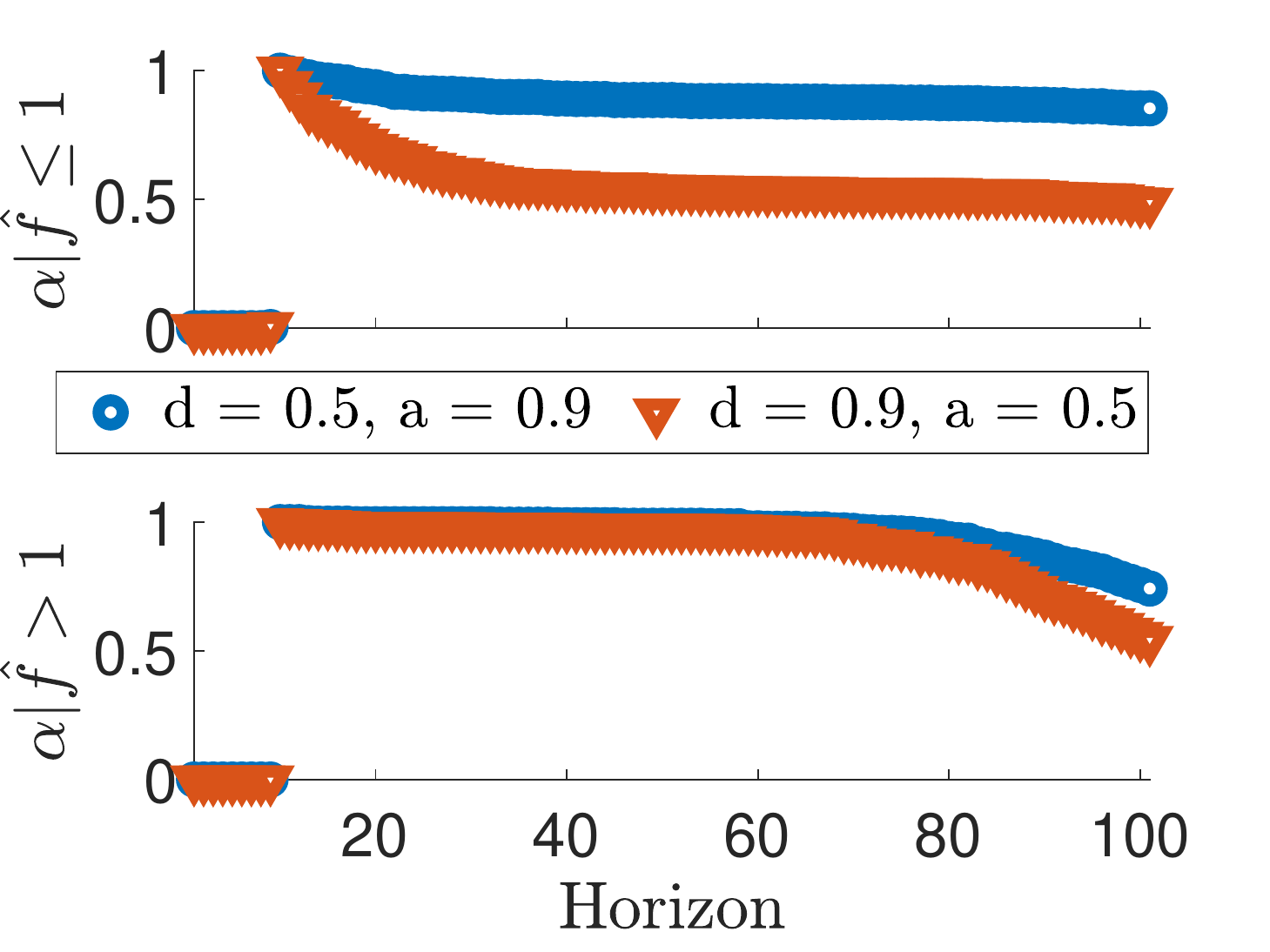}
			\label{fig:FD_recover}	
		}
		\caption{\small (a) \texttt{FlipDyn} state, $\alpha$ over the horizon length of $L = 100$ for an $n$-dimensional system with $\alpha_{0} = 0$. (b) \texttt{FlipDyn} state, $\alpha$ over the same horizon length of $L = 100$ for an $n$-dimensional system, with $\alpha_{0} = 0$ and $\alpha_{10} = 1$, i.e., an initial defender control and a sudden takeover by an adversary at time instant 10.}	
		\label{fig:FlipDyn_state}
	\end{center}
	\vspace{-0.25in}
\end{figure}

\section{Conclusion and Future Directions}\label{sec:Conclusion}
We introduced a resource takeover game between a defender and an adversary, in which the resource represents the control input signals of a dynamical system. We posed the takeover problem as a zero-sum two-player game over a finite time period, inspired by the well-studied FlipIT model. The payoffs for our \texttt{FlipDyn} game are modeled as state-dependent costs incurred by both the defender and adversary. We computed for the policy of each player, i.e., at what time instances should a player choose to takeover the resource. We derived the value of the physical state for a given \texttt{FlipDyn} state and the corresponding policy for any general system. In particular, we derived closed-form expressions for linear dynamical system leading to an exact value function computation for the 1-dimensional case, and an approximate value function for $n$-dimensional systems. Finally, we illustrate the results of the FlipDyn game on numerical examples and comment on the recovery of such a setup from loss of control.


Our current work relies on full state observability of even the \texttt{FlipDyn} state. In future works, we will assume only output feedback for inferring the \texttt{FlipDyn} state of the system. We also plan to include bounded process and measurement noise and evaluate its impact on the policy of the \texttt{FlipDyn} game. Finally, we will compare the existing and future solution against a learning-based method to explore which one would be a better fit for practical systems.

\addtolength{\textheight}{-12cm}   




\bibliographystyle{IEEEtran}
\bibliography{references}

\end{document}